\pdfoutput=1
\documentclass[11pt]{article}
\usepackage[
backend=biber,
style=alphabetic,
maxbibnames=15,
maxalphanames=10,
minalphanames=6,
doi=false,
isbn=false,
url=false,
eprint=false,
backref=true,
]{biblatex}
\DefineBibliographyStrings{english}{
  backrefpage  = {},
  backrefpages = {},
}

\DeclareFieldFormat{bracketswithperiod}{\mkbibbrackets{#1}}

\renewbibmacro*{pageref}{%
  \iflistundef{pageref}
    {}
    {\printtext[bracketswithperiod]{%
       \ifnumgreater{\value{pageref}}{1}
         {\bibstring{backrefpages}}
         {\bibstring{backrefpage}}%
       \printlist[pageref][-\value{listtotal}]{pageref}}%
     }
 }
     
\renewbibmacro{in:}{}
\AtEveryBibitem{\clearfield{pages}}

\usepackage[letterpaper,margin=1in]{geometry}
\usepackage{lipsum}

\newcommand\blfootnote[1]{%
  \begingroup
  \renewcommand\thefootnote{}\footnote{#1}%
  \addtocounter{footnote}{-1}%
  \endgroup
}

\usepackage[utf8]{inputenc} %
\usepackage[T1]{fontenc}    %

\usepackage{url}            %
\usepackage{booktabs}       %
\usepackage{nicefrac}       %
\usepackage{microtype}      %
\usepackage{enumitem}
\usepackage{dsfont}
\usepackage{multicol}
\usepackage{xcolor}
\usepackage{mdframed}

\usepackage{turnstile}
\usepackage{amsthm,amsfonts,amsmath,amssymb,epsfig,color,float,graphicx,verbatim, enumitem}

\usepackage{algorithmicx}
\usepackage[noend]{algpseudocode}
\usepackage{algorithm}

\usepackage{mathtools}
\usepackage{bbm}

\usepackage{soul}

\newlist{myEnumerate}{enumerate}{9}
\setlist[myEnumerate,1]{label=(\arabic*)}
\setlist[myEnumerate,2]{label=(\Roman*)}
\setlist[myEnumerate,3]{label=(\Alph*)}
\setlist[myEnumerate,4]{label=(\roman*)}
\setlist[myEnumerate,5]{label=(\alph*)}
\setlist[myEnumerate,6]{label=(\arabic*)}
\setlist[myEnumerate,7]{label=(\Roman*)}
\setlist[myEnumerate,8]{label=(\Alph*)}
\setlist[myEnumerate,9]{label=(\roman*)}

\usepackage{thm-restate}

\definecolor{green}{rgb}{0.0, 0.5, 0.0}
\usepackage[colorlinks,citecolor=blue,linkcolor=magenta,bookmarks=true,hypertexnames=false]{hyperref}
\usepackage[nameinlink,capitalize]{cleveref}

\crefformat{equation}{(#2#1#3)}
\crefname{lemma}{lemma}{lemmata}
\crefname{claim}{claim}{claims}
\crefname{theorem}{theorem}{theorems}
\crefname{proposition}{proposition}{propositions}
\crefname{corollary}{corollary}{corollaries}
\crefname{claim}{claim}{claims}
\crefname{remark}{remark}{remarks}
\crefname{definition}{definition}{definitions}
\crefname{fact}{fact}{facts}
\crefname{question}{question}{questions}
\crefname{condition}{condition}{conditions}
\crefname{algorithm}{algorithm}{algorithms}
\crefname{assumption}{assumption}{assumptions}
\crefname{notation}{notation}{notation}
\crefname{cond}{Condition}{Conditions}
\crefname{ineq}{Constraint}{Constraints}
\crefname{contModel}{Contamination Model}{Contamination Models}

  {%
   \par\noindent{\bfseries\upshape Proof Sketch\ }%
  }%
  {\qed}

\newtheorem{theorem}{Theorem}[section]
\newtheorem{lemma}[theorem]{Lemma}
\newtheorem{proposition}[theorem]{Proposition}
\newtheorem{corollary}[theorem]{Corollary}

\newtheorem{definition}[theorem]{Definition}
\newtheorem{contModel}{Contamination Model}

\newtheorem{fact}[theorem]{Fact}
\newtheorem{question}{Question}
\theoremstyle{definition}

\newtheorem{remark}[theorem]{Remark}
\newtheorem{example}[theorem]{Example}

\usepackage{fontawesome}

\newlist{itemizec}{itemize}{2}
\setlist[itemizec,1]{label=\faCaretRight ,wide, parsep= 0.05pt, left = 15pt}

\newcommand{\eps}{\epsilon}
\renewcommand{\tilde}{\widetilde}

\newcommand{\Ind}{\mathds{1}}
\newcommand{\1}{\Ind}

\newcommand{\E}{\operatorname*{\mathbf{E}}}
\newcommand{\pE}{\operatorname*{\widetilde{\mathbf{E}}}}

\newcommand{\poly}{\operatorname*{\mathrm{poly}}}

\renewcommand{\vec}[1]{\boldsymbol{\mathbf{#1}}}

\newcommand{\bP}{{\mathbb P}}
\newcommand{\trace}{\operatorname{tr}}
\newcommand{\rank}{\operatorname{rank}}

\def\P{\mathbb P}
\def\R{\mathbb R}

\def\N{\mathbb N}

\def\Z{\mathbb Z}

\newcommand\numberthis{\addtocounter{equation}{1}\tag{\theequation}}

\newcommand{\cA}{\mathcal{A}}
\newcommand{\cB}{\mathcal{B}}

\newcommand{\cD}{\mathcal{D}}
\newcommand{\cE}{\mathcal{E}}

\newcommand{\cH}{\mathcal{H}}
\newcommand{\cI}{\mathcal{I}}
\newcommand{\cJ}{\mathcal{J}}

\newcommand{\cM}{\mathcal{M}}
\newcommand{\cN}{\mathcal{N}}
\newcommand{\cO}{\mathcal{O}}
\newcommand{\cP}{\mathcal{P}}

\newcommand{\cS}{\mathcal{S}}

\newcommand{\cV}{\mathcal{V}}
\newcommand{\cW}{\mathcal{W}}

\newcommand{\bA}{\vec{A}}
\newcommand{\bB}{\vec{B}}

\newcommand{\bI}{\vec{I}}

\newcommand{\bM}{\vec{M}}

\newcommand{\bV}{\vec{V}}

\newcommand{\Set}[1]{\left\{#1\right\}}

\newcommand{\hide}[1]{}

\DeclareMathOperator*{\pr}{\mathbf{Pr}}
\newcommand{\eqdef}{\stackrel{{\mathrm {\footnotesize def}}}{=}}

\newcommand{\op}{\textnormal{op}}

\newcommand{\tr}{\mathrm{tr}}
\let\vec\mathbf

\newcommand{\weak}{\mathrm{weak}}
\newcommand{\strong}{\mathrm{strong}}
\newcommand{\bddcov}{\mathrm{cov}}
\newcommand{\bddmom}{{k-\mathrm{bdd-moments}}}
\newcommand{\bddsubg}{\mathrm{subgaussian}}

\allowdisplaybreaks

\title{Optimal Robust Estimation under Local and Global Corruptions:\\
Stronger Adversary and Smaller Error
}

\author{
Thanasis Pittas\\
University of Wisconsin-Madison\\
{\tt pittas@wisc.edu}\\
\and
Ankit Pensia\\
Simons Institute for the 
Theory of Computing\\
{\tt ankitp@berkeley.edu}
}

\begin{document}

\maketitle

\begin{abstract}
Algorithmic robust statistics has traditionally focused on the contamination model where a small fraction of the samples are arbitrarily corrupted.
We consider a recent contamination model
that combines two kinds of corruptions: (i) small fraction of arbitrary outliers, as in  classical robust statistics, and (ii) local perturbations, where samples may undergo bounded shifts on average. 
While each noise model is well understood individually, 
the combined contamination model poses new algorithmic challenges, with only partial results known. 
Existing efficient algorithms are limited in two ways: (i) they work only for a weak notion of local perturbations, and (ii) they obtain suboptimal error for isotropic subgaussian distributions (among others). 
The latter limitation led
\cite{NieGS24} to hypothesize that improving the error might, in fact, be computationally hard.
Perhaps surprisingly, we show that information theoretically optimal error can indeed be achieved in polynomial time, under an even \emph{stronger} local perturbation model  (the sliced-Wasserstein metric as opposed to the Wasserstein metric).
Notably, our analysis reveals that the entire family of stability-based robust mean estimators 
continues to work optimally in a black-box manner for the combined contamination model. 
This generalization is particularly useful in real-world scenarios where the specific form of data corruption is not known in advance.
We also present efficient algorithms for distribution learning and principal component analysis
in the combined contamination model.
\end{abstract}
\blfootnote{Authors are listed in random order.}



\section{Introduction}\label{sec:intro}

We study the problems of high-dimensional parameter estimation and distribution learning in the setting where the  data available to the statistician may be corrupted.
We start with the former problem, where
a prototypical task 
is that of (multivariate) mean estimation, defined as follows:
Let $\cP$ be a family of probability distributions over $\R^d$. Given a set of samples $S$ generated from an (unknown) distribution $P \in \cP$, compute (in a computationally-efficient manner) an estimate $\widehat{\mu}$ 
that is close to the true mean $\mu_P:= \E_{X \sim P}[X]$ in the Euclidean norm (with high probability).

The precise stochastic process generating the data $S$ plays a crucial role
in this problem.
Much of the historical focus has been on the \emph{ideal} setting, where $S$ is a set of i.i.d.\ samples from $P \in \mathcal{P}$.
However, this i.i.d.\ assumption is often violated in real-world scenarios, a phenomenon known as \textit{model misspecification}.
Two prominent examples of such misspecification include (i) outliers, which might occur due to 
data poisoning attacks~\cite{BarNJT10-security,BigNL12-poisoning,SteKL17-certified,TraLM18-backdoor,HayKSO21}, 
errors in the front-end systems in geometric  perception~\cite{YangCarl23}, biological settings~\cite{RosPWCKZF02-Gen,PasLJD10-ancestry,LiATSCRCBFCM08-science}, and (ii) local distribution shifts~\cite{ChaoDobriban23,YanZLL24-ood} due to varying biases of different sensors from which data is collected.
Importantly, in high dimensions, these  corruptions can lead to dramatic failures of standard off-the-shelf estimators (developed for the i.i.d.\ setting). 

The field of robust statistics was initiated in the 1960s to develop estimators that were robust to perturbations in the data~\cite{HubRon09,HamRRS11}.
Starting from the seminal work of Huber~\cite{Hub64},
the prototypical contamination model to capture these perturbations is when a small fraction of arbitrary outliers is included in the data, formalized as the adversary below. 
\begin{contModel}[Global contamination]
\label{def:outliers}
Let $\eps \in (0,1/2)$.
    Let $S$ be a multiset  of $n$ points in $\R^d$.
    Consider all $n$-sized sets in  $\R^d$ that differ in at most $\eps$-fraction of points, i.e.,
        $\cO(S,\eps) := \left\{S'\subset \R^d: |S'| =n \text{ and } |S' \cap S| \geq (1- \eps)n \right\}$.
    The adversary can return any set $T \in \cO(S,\eps) $.
    We call points in $S$ to be the inliers and points in $T\setminus S$ to be the outliers.
\end{contModel} 
These outliers are termed ``global'' because they can be completely arbitrary, without any constraints on their magnitude.
The effect of such outliers on the \emph{statistical} rates of estimation was understood early in the statistics literature~\cite{HubRon09,DonLiu88a}.
However, early robust statistics algorithms were either computationally infeasible in large dimensions, with runtimes  scaling exponentially with the dimension $d$, or achieved asymptotic errors scaling with $d$.\footnote{
An important distinction from classical statistics (with i.i.d.\ data) is that, due to biases caused by the outliers, the optimal asymptotic error (as the number of samples goes to infinity) may \textit{not} vanish.
Still, for many well-behaved distribution families,
the optimal asymptotic error is a function that depends only on the contamination rate $\eps$.
Importantly, the optimal asymptotic error is dimension-free for such families.}
In the last decade, the sub-field of algorithmic robust statistics has resolved this issue by developing a family of computationally-efficient algorithms for robust estimation under global outliers~\cite{DiaKKLMS16-focs,LaiRV16}.
We summarize these results for mean estimation below:

\begin{fact}[Algorithms for robust mean estimation; see, for example, the book \cite{DiaKan22-book}]
\label{fact:ars-existing}
    Let $\cP$ be one of the following distribution families:
    \begin{itemizec}
        \item (bounded covariance) $\cP_{\bddcov}$:  all distributions that have covariance at most identity.
        \item (isotropic bounded moments) $\cP_\bddmom$:  all distributions that have identity covariance and bounded $k$-th moments for an even $k \in \N$.\footnote{That is, $\E_{X \sim P}[|\langle v, X - \mu_P\rangle|^k]^{1/k} \leq \sigma_k$
        for a constant $\sigma_k$.\label{ftn:bdd}}
        \item (isotropic subgaussian) $P_{\bddsubg}$: all subgaussian distributions that have identity covariance.
    \end{itemizec}
    There exist \emph{polynomial-time} algorithms for robust mean estimation under \Cref{def:outliers} with (near)-optimal error for the distribution families above.
\end{fact}

\noindent We refer the reader to \Cref{sec:additional-prelims} for the exact rates and further discussion of their optimality.

\looseness=-1
While \Cref{def:outliers} has served as an extremely successful testbed for the development of robust statistics for around sixty decades, 
a potential weakness  
is that the bulk of the inliers, $S \cap T$, remain \emph{exactly} the same.
That is, it permits \emph{global} but \emph{sparse} perturbations---corresponding to a small (sparse) fraction of arbitrary (global) outliers. 
However, it does not account for \emph{dense} but \emph{local} perturbations.
Such dense, local perturbations could arise in practice, for example, when measurements made by differently miscalibrated sensors.
To model such local perturbations, recent works have  proposed using the Wasserstein distance ~\cite{ZhuJT20-gan,ZhuJS19,LiuLoh22,ChaoDobriban23,NieGS24}, which measures perturbations using the Euclidean norm.
In contrast, we consider a significantly strengthened adversary:    

\begin{contModel}[Strong local contamination]
\label{def:strongWass-intro}
Let $\rho \geq 0$.
    Let $S_0=\{x_1,\ldots,x_n \}$ be an $n$-sized set in $\R^d$. 
    Consider an adversary that perturbs each point $x_i$ to $\widetilde{x}_i$ with the only restriction that in each direction, the average perturbation (over $n$ points) is at most $\rho$. Formally, we define 
    \begin{align}
\cW^\strong(S_0,\rho) := \left\{S= \{\widetilde{x}_1,\dots, \widetilde{x}_n\} \subset \R^d :     \sup_{v \in \R^d: \|v\|_2=1}\,\,\frac{1}{n} \sum_{i \in [n]} \big|v^\top (\widetilde{x}_i - x_i) \big| \leq \rho    \right\}\,.
    \end{align}
    The adversary returns an arbitrary set  $S \in \cW^\strong(S_0,\rho)$ after possibly reordering the points.
\end{contModel}

\Cref{def:strongWass-intro} allows each inlier point to be shifted in a bounded manner in any direction.
Still, the overall perturbation in the Euclidean norm could be significant:
Taking
the simplest setting of Gaussian inliers $x_i$ and independent Gaussian perturbations $\widetilde{x}_i = x_i + \rho z_i$ for $z_i \sim \cN(0,\bI)$ as an example, we see that \emph{each} inlier is perturbed roughly by  $\Theta(\rho \sqrt{d})$ in the Euclidean norm, 
which is of the same order as the distance of the inlier from the mean $\mu$.\footnote{This follows by the fact that the mass of Gaussian distribution is extremely concentrated on a thin spherical shell of radius $\sqrt{d}$ around the mean.}
Despite the possibility of these large perturbations, one can see that \Cref{def:strongWass-intro}, in isolation, is rather benign: 
(i) the sample mean of the set $S$ incurs at most $O(\rho)$ error due to local perturbations, and
(ii) $\Omega(\rho)$ error is also information-theoretically unavoidable, since all points could be perturbed by $\rho$ in the same direction. 
Together, (i) and (ii) imply that simply outputting the sample mean of the contaminated set achieves error with optimal dependence on $\rho$.
Moreover, as we already saw in \Cref{fact:ars-existing}, mean estimation in \Cref{def:outliers} is already well-understood.
However, somewhat surprisingly, the contamination model combining the two we have seen above
poses new  algorithmic challenges.
\begin{contModel}[Global and strong local contamination]
\label{def:combined-outlier-local-intro}
Let $\eps \in (0,1/2)$ and $\rho > 0$.
Let $S_0 = \{x_1,\dots,x_n\}$ be a set of $n$ points in $\R^d$.
The adversary can return an arbitrary set $T$ such that $T \in \cO(S,\eps)$ for some $S \in \cW^\strong(S_0,\rho)$. 
\end{contModel}
The combined contamination model is more reasonable because data might be perturbed in numerous unforeseen ways~\cite{ZhuJT20-gan,NieGS24}.
To highlight the algorithmic challenges posed by the combined contamination model, observe that the two algorithms mentioned above run into problems: 
(i) the sample mean is highly sensitive to arbitrary outliers, and 
(ii) the algorithms from \Cref{fact:ars-existing} are no longer guaranteed to work because the local contamination destroys the moment structure. 
To elaborate the latter,  the prescribed recipe of algorithmic robust statistics relies crucially on the covariance of the inliers being bounded (more generally, bounded higher moments). 
In contrast, even though the local corruptions have bounded \emph{first} moment by definition, the \emph{higher moments} might be arbitrarily large (cf.\ \Cref{ex:local-corrution-cov}).
The aim of our work is to develop algorithms for fundamental estimation tasks that are robust to \Cref{def:combined-outlier-local-intro}. 
In particular, we focus on the tasks of mean estimation and distribution learning; we also show that our techniques extend to principal component analysis at the end.

\subsection{Motivating Questions}
\label{sec:questions}
In this section, we describe the questions that motivated our technical results in \Cref{sec:our-results}. 
For brevity, we focus only on the task of mean estimation in this section.

\looseness=-1Starting with the \emph{information-theoretic} error rates,  for any distribution family $\cP$ in \Cref{fact:ars-existing}, it is easy to see that
 the optimal error for mean estimation under \Cref{def:combined-outlier-local-intro} is $\Theta\left(f_\mathcal{P}^{\mathrm{robust}}(\eps) + \rho \right)$,
 where $f_\mathcal{P}^{\mathrm{robust}}(\eps)$ denotes the optimal asymptotic error for \Cref{def:outliers}.
This error is achieved by 
the multivariate trimmed mean, which 
is also optimal
in each model individually~\cite{LugMen21-trim}.
Unfortunately, the trimmed mean is not computationally efficient in high dimensions, with runtime scaling exponentially with the dimension,
leading to the central focus of our work:

\begin{question}
    \label{ques:efficient-strong}
    Can we perform robust estimation (mean estimation and  distribution learning) under \Cref{def:combined-outlier-local-intro} in a computationally-efficient manner?
\end{question}

\looseness=-1The simultaneous success of the multivariate trimmed mean 
across both global and local contamination models suggests an optimistic idea: perhaps the existing outlier-robust algorithms from \Cref{fact:ars-existing} continue to work for the more general \Cref{def:combined-outlier-local-intro}---despite this being far from obvious, as noted earlier.
The main contribution of our work is to show that this is indeed the case.

As we stress below, existing results for \Cref{ques:efficient-strong} achieve suboptimal error rates, both in terms of the global contamination parameter $\eps$ and the local contamination parameter $\rho$.

\paragraph{Dependence on Local Perturbation $\rho$.} Starting with the dependence on $\rho$, the 
most relevant work who studied efficient algorithms for \Cref{ques:efficient-strong} is
the recent paper of \citeauthor{NieGS24}~\cite{NieGS24}.
However, they used a weaker definition for the local corruptions, defiend next.
\begin{contModel}[Weak local contamination~\cite{NieGS24}]
    \label{def:weakWas}
    Let $\rho > 0$.
    Let $S_0=\{x_1,\ldots,x_n \}$ be a multiset  of $n$ points in $\R^d$. 
    Define $\cW^\weak(S_0,\rho)$ as follows:
    \begin{align*}
        \cW^\weak(S_0,\rho) := \Big\{S= \{\widetilde{x}_1,\dots, \widetilde{x}_n\} \subset \R^d:  \tfrac{1}{n} {\textstyle \sum_{i \in [n]}}\|\widetilde{x}_i - x_i\|_2 \leq \rho\Big\}.
    \end{align*}
    The adversary can return any $S \in \cW^\weak(S_0,\rho) $ after possibly reordering the samples.
\end{contModel}
With  the above definition of local perturbations (as opposed to \Cref{def:strongWass-intro}), 
\cite{NieGS24} focused on the following combined contamination model (as opposed to \Cref{def:combined-outlier-local-intro}).
\begin{contModel}[Global and weak local contamination \cite{NieGS24}]
\label{def:combined-outlier-local-weak-intro}
Let $\eps \in (0,1/2)$ and $\rho > 0$.
Let $S_0 = \{x_1,\dots,x_n\}$ be a set of $n$ points in $\R^d$.
The adversary can return any  set $T$ such that $T \in \cO(S,\eps )$ for some $S \in \cW^\weak(S_0,\rho)$. 
\end{contModel}

In the contamination model above, local perturbations have bounded average Euclidean norm, rather than being bounded in each direction individually as in \Cref{def:combined-outlier-local-intro}.
The observation that the analysis of trimmed mean  relies only on the boundedness in each direction individually (as opposed to bounded Euclidean norm) served as an inspiration for us to consider \Cref{def:combined-outlier-local-intro} and \Cref{ques:efficient-strong}.

Comparing the weak local contamination (\Cref{def:weakWas}) with the strong local contamination (\Cref{def:strongWass-intro}) quantitatively, we see that the corresponding contamination radii $\rho$ could greatly differ in high dimensions for a given perturbation. Formally, for any set $S_0 \subset \R^d$:
\begin{align}
\cW^\weak(S_0,\rho) \subset  \cW^\strong(S_0,\rho) \subseteq \cW^\weak(S_0,\rho \sqrt{d}),
\end{align}
with the $\sqrt{d}$ factor being necessary for the last inclusion.
Hence, applying the results of \cite{NieGS24} to our contamination model (\Cref{def:combined-outlier-local-intro}) leads to an error that has an extraneous $\sqrt{d}$ in front of $\rho$, which is undesirable in large dimensions.
More importantly, the technical arguments in  \cite{NieGS24} critically rely on the local perturbations having small Euclidean norms,
raising the question:
\begin{question}
    \label{ques:strong-vs-weak}
    Do weak and strong local contamination lead to different computational landscapes, when combined with global contamination? In particular, does the dependence on $\rho$ in the computationally-efficiently achievable error differ between \Cref{def:combined-outlier-local-intro,def:combined-outlier-local-weak-intro}? 
\end{question}

\paragraph{Dependence on Global Perturbation $\eps$.}Turning our attention to the dependence on $\eps$ (the fraction of global outliers), even for the weaker \Cref{def:combined-outlier-local-weak-intro},  \cite{NieGS24} achieves only partial results.
While the error of their algorithm has optimal $\eps$-dependence  for the family $\cP_{\bddcov}$ of bounded covariance distributions, defined in \Cref{fact:ars-existing}, 
the  dependence is suboptimal for the other two important distribution families $\cP_{\bddmom}$ and $\cP_{\bddsubg}$ from \Cref{fact:ars-existing}.\footnote{In fact, it can be shown that their algorithm  necessarily has suboptimal dependence on $\eps$ even in one dimension.}
In fact, \cite{NieGS24} conjectured that
no polynomial-time algorithm 
achieves optimal error for $\cP_{\mathrm{subgaussian}}$, noting that 
``\emph{We suspect that there may be similar obstacles 
[(computational hardness)]
as those known for robust mean estimation with stable but non-isotropic distributions}~\cite{HopLi19}''. 
The following question thus generalizes an  
open problem in \cite{NieGS24}:
\begin{question}
    \label{ques:local-optimal-eps}
    Do local corruptions (either weak or strong) induce new information-computation gaps for robust estimation? In particular, does the dependence on $\eps$ in the computationally-efficiently achievable error change in the presence of local contamination?
\end{question}

In our work, we completely resolve \Cref{ques:efficient-strong,ques:strong-vs-weak,ques:local-optimal-eps}. Importantly, our techniques generalize in a uniform manner for mean estimation, distribution learning, and principal component analysis.

\subsection{Our Results}
\label{sec:our-results}
We present efficient algorithms for mean estimation in \Cref{sec:mean-results} and distribution learning in \Cref{sec:distribution-learning-results}.
We defer the results for robust principal component analysis to \Cref{sec:pca}.

\subsubsection{Mean Estimation}\label{sec:mean-results}

We start with our results for the problem of  mean estimation. 
Our conceptual contribution is
to show that a 
family of existing  
outlier-robust
algorithms
continue to work under the combined contamination model.
In particular, we consider the \emph{stability-based algorithms} that are guaranteed to work for global contamination 
as long as the inliers satisfy the following deterministic condition called \emph{stability}.
For a set $S$, we use $\mu_S$ and $\vec \Sigma_S$ to denote the empirical mean and empirical covariance of $S$, respectively.

\begin{definition}[Stability, see, e.g., \cite{DiaKan22-book}] \label{def:stability1}
Let $\eps \in (0,1/2)$ and $\delta \in [\eps,\infty)$.
A finite multiset $S \subset \R^d$ is called $(\eps,\delta)$-stable with respect to $\mu \in \R^d$ if for
every $S' \subseteq S$ with $|S'| \geq (1-\eps)|S|$, the following  hold: 
\begin{itemizec}
    \item (Sample mean)  
    $\left\|\mu_{S'}- \mu \right\|_2 \leq \delta$,
    \item (Sample covariance)     $\left\|\vec\Sigma_{S'} - \bI_d \right\|_\op \leq \delta^2/\eps\,.$ 
\end{itemizec}
\end{definition}
The stability condition posits that for all large subsets, the sample mean is close to the true mean $\mu$, and the sample covariance is comparable to identity in the operator norm. It is known (\Cref{fact:stability-rates}) that stability holds with high probability for nice distribution families
such as the ones from \Cref{fact:ars-existing}, and a growing body of work has developed  algorithms for mean estimation with global outliers under the sole assumption that the inliers are stable. 
\begin{definition}[Stability-based Algorithms]\label{def:stability-based-alg}
Let $S$ be an $(\eps,\delta)$-stable set with respect to an (unknown) $\mu \in \R^d$.
Let $T$ be any set such that $T \in \cO(S,\eps)$ (cf. \Cref{def:outliers}).
We call an algorithm $\cA(T,\eps,\delta)$ stability-based algorithm if it takes as an input $T$, $\eps$, and $\delta$, and outputs an estimate $\widehat{\mu}$ in polynomial time such that, with high probability, 
$\|\widehat{\mu} -\mu\|_2 \lesssim \delta$.
\end{definition}
There exist such algorithms based on convex programming \cite{DiaKKLMS16-focs,SteCV18,CheDG19}, iterative filtering \cite{DiaKKLMS16-focs}, gradient descent \cite{CheDGS20,ZhuJS20}. 
Over the years, these stability-based algorithms have been optimized to be near-optimal in other important aspects: runtime
\cite{CheDG19,DonHL19,DepLec22}, sample complexity \cite{DiaKP20},  and memory \cite{DiaKPP22-streaming}. 
Stability-based algorithms give a unified way to achieve the rates in \Cref{fact:ars-existing}.
In fact, they have found black-box consequences on other problems such as principal component analysis (\Cref{sec:pca}) and linear regression~\cite{PenJL20}. 
Our main result allows us to seamlessly apply this huge repertoire of algorithms to the combination of global and local contamination.

\begin{restatable}[Main Result for Mean Estimation]{theorem}{ThmMeanEstStability}
    \label{thm:mean-estimation}
    Let $c$ be a sufficiently small positive constant and $C$ a sufficiently large constant.
    Let outlier rate $\eps \in (0,c)$ and  contamination radius $\rho > 0$.
    Let $S_0$ be a set that is $(\eps,\delta)$-stable
    with respect to an (unknown) $\mu \in \R^d$, where $\delta > \eps$. 
    Let $T$ be a corrupted dataset after  $\eps$-fraction of global outliers and $\rho$-strong local corruptions (as per \Cref{def:combined-outlier-local-intro}).
    Then, any stability-based algorithm $\cA(T,\eps,\tilde \delta)$ executed with input $T, \eps, \tilde \delta = C \cdot (\delta+\rho)$, outputs an estimate  $\widehat{\mu}$ such that with high probability (over the internal randomness of the algorithm):
        $\|\widehat{\mu} - \mu\|_2 \lesssim \delta + \rho$.
\end{restatable}
Observe that the dependence on $\rho$ is optimal.
An astute reader might note that the success of stability-based algorithms is somewhat surprising in the light of the following contradictory facts: (i) stability requires the covariance to be bounded and (ii) the covariance of a (stable) set $S_0$ can increase drastically after local contamination, say, $S \in \cW^\strong(S_0,\rho)$.
One of our main technical results shows that even though $S$ is not stable (due to large covariance), it does contain a large stable set, which suffices for stability-based algorithms to work.

Combining this with the fact that stability holds with high probability for the distribution families of interest (\Cref{fact:stability-rates}),
we obtain the following corollary.

\begin{corollary}\label{cor:mean_est}
\looseness=-1
    Let $\cP$ be a family of distributions.
    Fix a $P \in \cP$ with the (unknown) mean $\mu$,
    Let $S_0$ be a set of $n$ i.i.d.\ samples from $P$.
    Let $T$ be a corrupted version of $S_0$ with local contamination parameter $\rho$ and global contamination rate $\eps$ (\Cref{def:combined-outlier-local-intro}).
    Let $\tau $ be the failure probability such that $\log(1/\tau)/n$ is less than an absolute constant.
    There exist computationally-efficient algorithms that (i) take as input $T$, $
    \eps$, $\cP$, $\tau$, and $\rho$ and (ii) output $\widehat{\mu}$ that satisfies the following guarantees with probability $1-\tau$:
    \begin{itemizec}
        \item If $\mathcal{P}$ is the family of isotropic subgaussian distributions, then
        $\|\widehat{\mu}-\mu\|_2 \lesssim \eps \sqrt{\log(1/\eps)} + \rho + \sqrt{d/n} + \sqrt{\log(1/\tau)/n}$.
        \item If $\mathcal{P}$ is the family of distributions with isotropic covariance and bounded $k$-th moments, then
        $\|\widehat{\mu}-\mu\|_2 \lesssim \eps^{1 - \frac{1}{k}}+ \rho+ \sqrt{(d \log d)/n} + \sqrt{\log(1/\tau)/n}$.
        \item If $\mathcal{P}$ is the family of distributions with covariance $\vec \Sigma\preceq \vec I$
        then
        $\|\widehat{\mu}-\mu\|_2 \lesssim \sqrt{\eps} + \rho+ \sqrt{(d \log d)/n} + \sqrt{\log(1/\tau)/n}$.
    \end{itemizec}
\end{corollary}
We highlight that the error rates above are known to be information-theoretically optimal in all the parameters $\eps, \rho, d, n, \tau $ (up to the $\sqrt{\log d}$ factor in the term $\sqrt{(d \log d)/n}$); See 
\Cref{sec:additional-prelims} for further discussion on optimality.
Therefore, \Cref{cor:mean_est}  
simultaneously answers \Cref{ques:efficient-strong,ques:strong-vs-weak,ques:local-optimal-eps} for robust mean estimation. 
In particular,  (i) both weak and strong local contamination yield the same computationally-efficient rates, answering \Cref{ques:strong-vs-weak} and (ii) local contamination (whether weak or strong) does not induce new information-computation gap, refuting the hypothesis in \cite{NieGS24} and answering \Cref{ques:local-optimal-eps}.

A particularly appealing property of our results is that a well-understood, practical, and highly-optimized family of algorithms
is shown to work against local corruptions as well. Furthermore, the fact that the same algorithms work for all contamination models (local, global, and combined), without any modification in the algorithm (other than simple modification in a single parameter), is highly advantageous, since in practice we often do not know in advance what kinds of corruptions are present in the dataset.\footnote{Although \Cref{thm:mean-estimation} indicates that $\rho$ is included in the input as part of the parameter $\tilde \delta = \rho + \delta$ for stability-based algorithms, there exist stability-based algorithms that do not require $\tilde \delta$ as an input parameter~\cite[Appendix A]{DiaKP20}.}

\looseness=-1Beyond the distribution families in \Cref{fact:ars-existing}, another important distribution family for which efficient algorithms are known is the class of distributions with certifiably bounded moments, which is the class of distributions for which the bounded moment conditions has a low-degree sum of squares proof.
These algorithms have the benefit that they do not need to know the covariance of the underlying inlier distribution.
We refer the reader to \Cref{sec:sos} for the related definitions and background.  
\begin{theorem}[Optimal Asymptotic Error for Certifiably Bounded Distributions; informal]
\label{thm:sos-mean}
Let $\eps \in (0,c)$ for a sufficiently small absolute constant $c$.
    Let $P$ be a distribution over $\R^d$ with mean $\mu$ and  $t$-th moment certifiably bounded by $M$.
    Then there is an algorithm that takes $n=\poly(d^t,1/\eps)$ samples, runs in time $\poly(n^t,d^{t^2})$,
    and outputs an estimate $\widehat{\mu} \in \R^d$ such that with high constant probability $\|\widehat{\mu} - \mu\|_2 \lesssim M \eps^{1-\frac{1}{t}}$.
\end{theorem}
The asymptotic error $\eps^{1-1/t}$ is again optimal for this class of distributions, and the sample complexity is qualitatively optimal for a broad family of algorithms such as statistical query algorithms and low-degree polynomials~\cite{DiaKKPP22-colt}.

\subsubsection{Distribution Learning}\label{sec:distribution-learning-results}
We now move beyond mean estimation to the problem of distribution learning with respect to the (sliced)-Wasserstein metric, defined below.
\begin{definition}[Sliced Wasserstein Distance]\label{def:sliced-Wasserstein}
Let $P,Q$ be two distributions. The $k$-sliced $p$-Wasserstein Distance is defined as follows:
\begin{align*}
    W_{p,k}(P,Q) :=\max_{\substack{\text{$\bV:$ \text{rank-$k$} } \\ \text{projection matrix}}} \inf_{\pi \in  \Pi(P,Q)} \E_{(x,x') \sim \pi}\left[ \|\bV( x-x')\|_2^p \right]^{1/p}\;,
\end{align*}
    where $\Pi(P,Q)$ is the set of all couplings of $P$ and $Q$. By slightly overloading our notation, when $S$ and $\hat{S}$ are sets of points, we denote by $W_{p,k}(S,\hat{S})$ the $k$-sliced $p$-Wasserstein distance between the uniform distributions over $S$ and $\hat{S}$, respectively.
\end{definition}
Distribution learning in the sliced Wasserstein metric has applications in distributionally robust optimization~\cite{NieGS24}.
To present our result in generality for this problem, we consider the following contamination model, that interpolates between \Cref{def:strongWass-intro,def:weakWas} for $k=1$ and $k=d$, respectively.

\begin{contModel}[Strong Wasserstein Contamination: Contamination in $W_{1,k}$]
\label{def:strongWass-k}
Let $\rho > 0$.
    Let $S_0=\{x_1,\ldots,x_n \}$ be a multiset of $n$ points in $\R^d$. 
    Define the following set of local perturbations which are small in all rank-$k$ subspaces:
    \begin{align}
\cW^\strong_{1,k}(S_0,\rho) := \Big\{S= \{\widetilde{x}_1,\dots, \widetilde{x}_n\} \subset \R^d :     \max_{\substack{\text{$\bV:$ rank-$k$ } \\ \text{projection matrix}}}\,\,\tfrac{1}{n} \sum_{i \in [n]} \|\bV (\widetilde{x}_i - x_i) \|_2 \leq \rho.    \Big\}
    \end{align}
    The adversary returns any set  $S \in \cW^\strong_{1,k}(S_0,\rho)$ after possibly reordering the points.
    We call the set $S$ to be a $\rho$-contaminated version of the set $S_0$ under the $k$-sliced Wasserstein adversary.
\end{contModel}
Observe that $\cW^\strong_{1,k'}$ is a stronger adversary than $\cW^\strong_{1,k}$ for $k' \leq k$.
The distribution problem we consider is the following:
Let $S_0$ be a stable set corresponding to the inliers, which is then corrupted by a combination of local corruptions (\Cref{def:strongWass-k}) and global corruptions (\Cref{def:outliers}), and the statistician's goal is to output a distribution $\hat{P}$ which is close to the uniform distribution over $S_0$ with respect to the $W_{1,k}$ metric (this is a natural metric to use for measuring the error since the local corruptions in \Cref{def:strongWass-k} are also measured using the same metric). Formally, our result is the following: 

\begin{restatable}[Main Result for Distribution Learning]{theorem}{MAINTHEOREM}\label{thm:main_distr_learning}
    Let $\eps \in (0,c)$ be a parameter for the outlier rate, where $c$ is a sufficiently small absolute constant, $\rho > 0$ be a parameter for the local contamination radius, and $\delta>\eps$ be a parameter for stability.
    Let $S_0$ be a set that is $(\eps,\delta)$-stable with respect to an (unknown) $\mu \in \R^d$. 
     For a slicing parameter $k \in [d]$, 
    let $T$ be the corrupted dataset after local and global corruptions from \Cref{def:strongWass-k,def:outliers} with parameters $\rho$ and $\eps$, respectively; formally, $T \in \cO(S,\eps)$ for some $S \in \cW^\strong_{1,k}(S_0,\rho)$.
    Then, there exists a polynomial-time algorithm that on input $T,\eps, \rho, \delta$, and $k' \in [k]$,
    outputs an estimate $\widehat{S} \subset T$ such that, with high constant probability, it holds that 
        $W_{1,k'} (\widehat{S}, S_0) \lesssim \delta \sqrt{k'} + \rho$. 
\end{restatable}
We note that \cite{NieGS24} also provides rates for distribution estimation in the $W_{1,k}$ metric.
However, their adversary for the local contamination is much weaker ($W^\strong_{1,d}$).
In contrast, our rates for both the local perturbation and accuracy are measured in $W_{1,k}$.

    Our result for mean estimation, \Cref{thm:mean-estimation}, is a special case of the above for $k'=1$.\footnote{It follows by a standard property of sliced-Wasserstein distance that $\|\mu_{\hat{S}} - \mu_{S_0}\|_2 \lesssim W_{1,1}(\hat S,S_0)$.}  
    For general $k\geq 1$, our proposed algorithm is a filtering-based algorithm that uses a certificate lemma from \citeauthor{NieGS24}~\cite{NieGS24}.
    We directly optimize this certificate in an efficient manner to obtain the optimal error $\delta \sqrt{k} + \rho$; In contrast, \cite{NieGS24} optimizes an approximation of their certificate, which leads to the larger error $\max(\delta, \sqrt{\eps}) \sqrt{k}  + \rho$.\footnote{Recall that for nice distribution families, $\delta$ is the function of $\eps$ from \Cref{fact:stability-rates}; importantly $\delta = o(\sqrt{\eps})$ for Gaussians and distributions with $k>2$ bounded moments}

    We now discuss the error guarantee of \Cref{thm:main_distr_learning} in more detail. 
    First, the error of  $\Omega(\rho)$ is trivially needed because each point could be shifted by distance $\rho$ along the same direction.
    Next, the error term $\delta \sqrt{k}$ is also optimal; see
    \cite[Corollary 5]{NieGS24}.

    We now show how to instantiate the error  guarantee of \Cref{thm:main_distr_learning} for learning a distribution $P$ over $\R^d$ as opposed to the uniform distribution over a set $S_0$.
    Applying the triangle inequality, we get a simple approximation $W_{1,k}(\hat S, P) \leq W_{1,k}(\hat S, S_0) + W_{1,k}(S_0, P) = O(\delta \sqrt{k} + \rho + W_{1,k}(S_0, P))$.
    While the first two terms are optimal (as shown above),
    the third term can be upper bounded by $\tilde O(\sqrt{d}kn^{-\frac{1}{\max(k,2)}})$ using \cite{Boedihardjo24}.
    On the other hand, it has been shown in \cite{NilesweedRigollet22} that even for clean i.i.d.\ data, any estimator $\hat P$ necessarily incurs error $W_{1,k}(\hat P, P)=\Omega(c_d n^{-1/\max(k,2)} + \sqrt{d/n})$ for a dimension-dependent term $c_d$.
    Thus, the resulting error guarantee 
    is tight up to the suboptimality of $W_{1,k}(S_0, P)$, which we leave for future work.

Finally, the results for robust principal component analysis are deferred to \Cref{sec:pca}.

\subsection{Related Work}

Our work lies broadly in the field of algorithmic robust statistics. 
We refer the reader to \citeauthor{DiaKan22-book}~\cite{DiaKan22-book} for a recent book on the topic. Within this line of our work, our work is most closely related to \cite{SteCV18,DepLec22,DiaKP20}, which we discuss in \Cref{sec:techniques}.
As mentioned earlier, robust statistics has primarily focused on \Cref{def:outliers}.
Some notable exceptions include \cite{ZhuJT20-gan,ZhuJS19, LiuLoh22, ChaoDobriban23,NieGS24}, discussed below in detail.

\citeauthor{ZhuJS19}~\cite{ZhuJS19} and \citeauthor{LiuLoh22}~\cite{LiuLoh22} studied the problems of covariance estimation and linear regression under the Wasserstein-1 perturbations.
Similarly, \citeauthor{ChaoDobriban23}~\cite{ChaoDobriban23} investigated the Wasserstein-$2$ perturbations and developed minmax-optimal estimators under those models. 
To the best of our knowledge, the combined contamination model (\Cref{def:combined-outlier-local-weak-intro}) was first proposed and studied in \citeauthor{ZhuJT20-gan}~\cite{ZhuJT20-gan}, inspired by chained perturbations in the computer vision literature~\cite{BorPD18}. \citeauthor{LiuLoh22}~\cite{LiuLoh22} also studied \Cref{def:combined-outlier-local-weak-intro} by focusing on the problems of covariance estimation and linear regression.
However, all of these works focused on the statistical aspects (and \emph{weak} local contamination), and did not provide computationally-efficient algorithms.

Our work is most closely related to \citeauthor{NieGS24}~\cite{NieGS24} who developed computationally-efficient algorithms for the combined contamination model (with weak local perturbations) in \Cref{def:combined-outlier-local-weak-intro}.
In contrast, we study the stronger \Cref{def:combined-outlier-local-intro}. 
We also obtain the improved dependence on $\eps$ for certain distribution families, which was phrased as an open question in their work.

Finally, we mention related works from the theory of optimal transport.
In fact, the combined contamination model (\Cref{def:combined-outlier-local-weak-intro}) is closely related to the notion of \emph{outlier-robust} optimal transport cost~\cite{BalCF20-robust-ot,NieGC22-duality}, but their focus is rather different. 
In fact, our results can be seen as learning when the samples are perturbed in outlier-robust \emph{sliced} optimal transport cost.
The sliced-Wasserstein distance has been studied in several recent works because it avoids the curse of dimensionality fundamental to the usual Wasserstein distance~\cite{RabPDB12,NadDCKSS20,ManBW22,Boedihardjo24,CheNR24-ot}.

\subsection{Overview of Techniques}\label{sec:techniques}

In this section, we give an overview of the challenges posed by the (strong) local contamination and highlight a key technical result towards establishing \Cref{thm:mean-estimation}.

We begin by highlighting the issue that local perturbations (whether strong or weak, as defined in \Cref{def:strongWass-intro,def:weakWas}, respectively) can significantly increase the covariance—or more generally, the higher moments—of the data. 

\begin{example}[Local Corruptions Can Destroy Higher Moment Structure]
\label{ex:local-corrution-cov}
Suppose the inliers are $S_0= \{x_1,\dots,x_n\}$ and they have identity covariance.
For a unit vector $v$, consider the following locally corrupted set $S = \{x_1+0.5\rho n v, x_2-0.5\rho n v,x_3,x_4,\dots,x_n\}$. 
These local corruptions increase the covariance of $S$ by at least $\Theta(\rho^2 n)$ in the operator norm.
In particular, the set $S$ is not stable since its stability parameter diverges with $n$.
\end{example}
While $S$ above does not have bounded covariance (and hence not stable), we see that only a tiny fraction of points contributes disproportionately to the covariance, and hence, we might as well consider them outliers since the stability-based algorithms are robust to outliers.
Thus, the goal shifts towards establishing the stability of a \emph{large subset} $S' \subset S$ of the locally perturbed data, which  would directly imply our result for mean estimation. 
More generally, for our distribution learning result, we need an analogous claim for a generalized notion of stability, defined below (where the notation $\mu_S$ denotes the empirical mean $\tfrac{1}{|S|}\sum_{x \in S}x$ and $\vec \Sigma_{S}$ the empirical second moment centered around $\mu$, i.e., $\tfrac{1}{|S|}\sum_{x \in S}(x-\mu)(x-\mu)^\top$):

\begin{restatable}[Generalized Stability]{definition}{GENERALIZEDSTABILITY}
\label{def:condition}
    Let $\eps\in(0,1/2)$ and $\delta \in [\eps,\infty)$. 
    Let $S$ be a set of points in $\R^d$ and $\mu$ be a vector in $\R^d$.
    We say that $S$ satisfies the $(\eps,\delta,k)$-generalized-stability with respect to $\mu$ if for all $S' \subseteq S$ with $|S'| \geq (1-\eps) S$, the following hold:
    \begin{itemizec}
        \item $\left\|\mu_{S'}-\mu \right\|_{2} \leq  \delta$.
        \item For every $\bV \in \cV_k$, $\left|  \left\langle \bV , \overline{\vec \Sigma}_{S'} - \bI \right\rangle \right| \leq  \delta^2/\eps$.
    \end{itemizec}
    where $\cV_k$ denotes the set of all rank-$k$ projection matrices.
\end{restatable}

As highlighted above, the key difficulty in establishing the generalized stability property of $S$ concerns the second property in \Cref{def:condition},
which posits that for all large subsets $S'$, 
the covariance is small in the sense that its inner product with any rank-$k$ projection  is at most $\delta^2/\eps$.
To simplify the discussion, let us demonstrate our ideas for the special case when $S’ = S$, i.e., the complete set.
The variance-like quantity $\langle \bV , \overline{\vec \Sigma}_{S} \rangle$ is mainly composed of two terms (ignoring the cross terms): (i) the covariance of the unperturbed data $S_0$, $\langle \bV , \overline{\vec \Sigma}_{S_0} \rangle$ and (ii) the second moment of the local perturbations: $\tfrac{1}{n}\sum_i  \|\bV \Delta_i\|_2^2$. 
The first term can be handled by stability of the original data.
 Thus, the goal is to find a large subset of local perturbations that have bounded second moment (\Cref{lem:existence-of-stability-and-W_1-and-W-2-new}).
To be more precise, we need to identify a $(1-\eps)$-fraction of local perturbations $\{ \Delta_i \}_{i \in [n]}$ with second moment matrix bounded by $\rho^2/\eps$ in every rank-$k$ projection.
However, establishing the existence of a large stable subset is significantly different for the weak and strong local contamination, as explained next.

\paragraph{Differences between weak and strong local contamination.} To highlight the challenges between the strong and weak local contamination, 
we define
$(\Delta_i)_{i \in [n]}$ and $(\Delta'_i)_{i \in [n]}$ to be $2n$ vectors in $\R^d$ corresponding to strong and weak local contamination, respectively. That is, these vectors satisfy
\begin{align}
    \label{eq:local-perturbation-intro}
    \sup_{\bV \in \cV_k} \frac{1}{n} \sum_{i \in [n]} \| \bV \Delta_i \|_2 \leq \rho\,\,\, \text{and}\,\,\,  \frac{1}{n} \sum_{i \in [n]} \| \Delta_i'\|_2 \leq \rho,
\end{align}
respectively, where $\cV_k$ denotes the set of all rank-$k$ projection matrices.
For the weak local contamination, finding a large stable subset of the $\{\Delta_i'\}_{i \in [n]}$ with bounded second moment is rather easy: Let $\cI \subset [n]$ be the set of indices corresponding to the  $(1-\eps)n$ many vectors from $\{\Delta_i'\}_{i \in [n]}$ with the smallest Euclidean norms.
It can be then checked that the $\{ \Delta'_i\}_{i \in \cI}$ have appropriately bounded second moment as follows:
For any $\bV \in \cV_k$, it holds 
\begin{align}
    \label{eq:truncation-intro-weak}
    \frac{1}{|\cI|} \sum_{i \in \cI} \Delta_i^\top \bV \Delta_i \leq  \frac{1}{|\cI|} \sum_{i \in \cI} \|\Delta_i\|_2^2
    \leq
     \max_{i \in \cI}\|\Delta'_i\|_2 \cdot \frac{1}{|\cI|} \sum_{i \in \cI} \| \Delta'_i\|_2
    \lesssim \frac{\rho}{\eps} \cdot \rho \lesssim \frac{\rho^2}{\eps},
\end{align}
where we used the Markov inequality to get that all $\Delta'_i$'s in $\cI$ have Euclidean norm at most $\rho/\eps$, and we also used that $ \tfrac{1}{|\cI|} \sum_{i \in \cI} \| \Delta'_i\|_2 \leq \rho$ by definition of the weak local perturbations.
Implicitly, this is the strategy used in \citeauthor{NieGS24}~\cite{NieGS24}.\footnote{While \cite{NieGS24} does not obtain the optimal dependence on the stability parameter after this truncation, a careful calculation leads to optimal stability parameter for the locally perturbed data; see \Cref{lem:existence-of-stability-and-W_1-and-W-2-new}. }

\looseness=-1However, this norm-based truncation  can not work for the strong local contamination. 
This is simply because $\|\Delta'_i\|_2$ might be $\Theta(\rho \sqrt{d/k})$
for all $i \in [n]$, and hence the resulting inequality in \Cref{eq:truncation-intro-weak} is too loose. 

\paragraph{Towards tackling strong local contamination.}
Let $\{\Delta_i\}_{i \in [n]}$ now be perturbations according to the strong local contamination (i.e., satisfying the first inequality in \eqref{eq:local-perturbation-intro}).
A natural strategy is to adopt the proof strategy in a \emph{direction-dependent} manner.
For a ``direction'' $\bV \in \cV_k$, we can define the set $\cI_\bV \subset [n]$ to be the set of $(1-\eps)n$ many indices with the smallest $\|\bV \Delta_i\|_2$'s.
A similar application of Markov's inequality implies that $\max_{i \in \cI_\bV} \|\bV \Delta_i\|_2^2 \leq \rho/\eps$.
Following the arguments similar to \Cref{eq:truncation-intro-weak}, we find that for any  $\bV \in \cV_k$:
\begin{align}
    \label{eq:truncation-intro-strong}
    \frac{1}{|\cI_\bV|} \sum_{i \in \cI_\bV} \| \bV \Delta_i \|_2^2 \lesssim \frac{\rho^2}{\eps}.
\end{align}
That is, for any  $\bV \in \cV_k$, there is a large subset whose second moment in the ``direction`` $\bV \in \cV_k$ is at most $\rho^2/\eps$.
However, the order of quantifiers of $\bV$ and the $\cI_\bV$ is reversed compared to what we want; we would like to find a single subset that works for every  $\bV$.
In what follows, we show that the order of quantifiers can actually be fixed
by establishing the following statement in this section:
\begin{proposition}
\label{prop:subset-covariance-intro}
    Let points $\Delta_i \in \R^d$ as in \Cref{eq:local-perturbation-intro}. 
    Then for every $\eps \in (0,1)$ there exists a subset $\cI \subseteq [n]$ such that (i) $|\cI| \geq (1-\eps)n$ and 
    (ii) for all $\bV \in \cV_k$,     
    $\frac{1}{|\cI|}\sum_{i \in \cI} \|\bV \Delta_i\|_2^2 \lesssim \rho^2/\eps$.
\end{proposition}
We note that the proposition above is deterministic. 
 Our proof strategy builds on \citeauthor{SteCV18}~\cite{SteCV18} and \citeauthor{DiaKP20}~\cite{DiaKP20}, with crucial differences, as explained next.
\cite{SteCV18} includes a similar result for the $k=1$ case, but their formulation and proof do not seem to capture the rank-$k$ sliced distance, and hence their result does not yield the more general version of $k\gg 1$, which is crucially needed for our result on distribution learning.
On the other hand,  \cite{DiaKP20} focuses on establishing good sample complexity (again for $k=1$) for stability (\Cref{fact:stability-rates}) as opposed to the deterministic statement above.

We now sketch the proof of \Cref{prop:subset-covariance-intro}. Instead of solving the discrete problem above (optimizing over all large subsets $\cI$), 
following \cite{SteCV18,DiaKP20},
we begin by performing a convex relaxation and define
\begin{align*}
    \Delta_{n,\eps} := \left\{w \in \R_+^n: \sum_{i=1}^n w_i = 1 \, ; 0 \leq  w_i \leq \frac{1}{(1-\eps)n}\,\,  \right\}\,.
\end{align*}
\looseness=-1A rounding argument shows that finding a  $w \in \Delta_{n,\eps}$ suffices to prove \Cref{prop:subset-covariance-intro},  i.e., 
it suffices to show that   
$\min_{w \in \Delta_{n,\eps}} \max_{\bV \in \cV_k} \sum_i w_i \| \bV \Delta_i \|_2^2 \lesssim \rho^2/\eps$ (\Cref{lem:rounding-cont-discrete} from \cite{DiaKP20}).
As alluded to earlier, if the order of quantifiers for $w$ and $\bV$ was reversed, the desired conclusion would follow from \Cref{eq:truncation-intro-strong}.

In order to reverse these quantifiers, we shall use the min-max duality for bilinear programs over convex compact sets.
Thus, we perform a convexification of the max variable and define $\cM_k := \{\bM \in \R^{d \times d}: 0 \preceq \bM \preceq \bI; \trace(\bM) = k\}$, which is the convex hull of $\{\bV \in \R^{d \times d}: \bV \in \cV_k\}$.
We thus arrive at the key reformulation:
\begin{align*}
    \min_{w \in \Delta_{n,\eps}} \max_{\bV \in \cV_k} \sum_{i=1}^n w_i \|\bV \Delta_i\|_2^2 
    = \min_{w \in \Delta_{n,\eps}} \max_{\bM \in \cM_k} \sum_{i=1}^n w_i  \Delta_i^\top \bM \Delta_i 
    &=  \max_{\bM \in \cM_k}\min_{w \in \Delta_{n,\eps}} \sum_{i=1}^n w_i  \Delta_i^\top \bM \Delta_i,
\numberthis \label{eq:min-max-M-intro}
\end{align*}
where the last equality is due to the min-max duality.

While the order of quantifiers allows us to perform direction-dependent truncation,
we are faced with the new challenge that we do not have guarantees on the behavior of $\{\Delta_i^\top \bM \Delta_i\}_{i=1}^n$ for a general $\bM \in \cM_k$.
To be precise, while \Cref{eq:truncation-intro-strong} implies that $\max_{\bV \in \cV_k}\min_{w \in \Delta_{n,\eps}} \sum_i w_i  \Delta_i^\top \bV \Delta_i  \lesssim \rho^2/\eps$, the same argument does not apply when the $\max$ is taken over $\cM_k$.
This is because of the non-linearity induced by the $\min_{w}$ operator (if it was linear in $\bM$, then the maximum over $\cV_k$ and the analogous maximum over $\bM \in \cM_k$ would have been equal by convexity).

A simple observation here is to note that if the $\Delta_i$'s were \emph{well-behaved} with respect to $\bM$ in the sense
that $\sup_{\bM \in \cM_k} \sum_{i=1}^n \frac{1}{n} \sqrt{\Delta_i^\top \bM \Delta_i} \lesssim \rho$, 
then we can control the right hand side in \Cref{eq:min-max-M-intro} by $\rho^2/\eps$ using the 
same truncation strategy as in \Cref{eq:truncation-intro-strong}.
Using a Gaussian rounding scheme, 
inspired by a similar rounding scheme from \citeauthor{DepLec22}~\cite{DepLec22}, 
we prove in \Cref{lem:average_roots} that 
\begin{align*}
    \sup_{\bM \in \cM_k} \sum_{i=1}^n \frac{1}{n} \sqrt{\Delta_i^\top \bM \Delta_i} \lesssim \sup_{\bV \in \cV_k} \frac{1}{n} \sum_{i \in [n]} \|\bV \Delta_i\|_2,
\end{align*}
which completes the proof of \Cref{prop:subset-covariance-intro}.

\paragraph{Completing the proof of \Cref{thm:mean-estimation}.}
In \Cref{lem:existence-of-stability-and-W_1-and-W-2-new}, we show that if the local perturbations have the second moment matrix bounded by $\rho^2/\eps$ in each ``direction'' $\bM \in \cM_k$, 
then these perturbations can degrade the stability parameter $\delta$ by at most additive $\rho$ (up to additional constant prefactors).
Importantly, \Cref{lem:existence-of-stability-and-W_1-and-W-2-new} preserves the dependence on $\delta$ as opposed to the analysis in \cite{NieGS24}, which obtains a bound in terms of  $\max(\delta, \sqrt{\eps})$.
\Cref{prop:subset-covariance-intro} implies that a large subset of local perturbations has bounded second moment matrix in each ``direction'' $\bV \in \cV_k$ (and by convexity the same is true for every direction $\bM \in \cM_k$).
Combining these two claims, we get the existence of a large stable subset after strong local contamination, finishing the proof of \Cref{thm:mean-estimation}.

\subsection{Organization}
The rest of the paper is organized as follows:
\Cref{sec:prelim} contains basic definitions and the key properties of stable sets that will be useful later on.
\Cref{sec:rounding} states the relationship between the low-rank projections and their convex counterparts.
In \Cref{sec:stability-under-local}, we show that local perturbations with bounded covariance suffice for stability.
\Cref{sec:mean-estimation,sec:distr_learning,sec:sos} include the proofs of \Cref{thm:mean-estimation,thm:main_distr_learning,thm:sos-mean}, respectively.
We include the results for principal component analysis in \Cref{sec:pca}.

\section{Preliminaries}
\label{sec:prelim}
\paragraph{Basic notation.}
We use $\mathbb{Z}_+$ for the set of positive integers and $[n]$ to denote $\{1,\ldots,n\}$. For a vector $x$ we denote by $\|x\|_2$ its  Euclidean norm. Let $\bI_d$  denote the $d\times d$ identity matrix (omitting the subscript when it is clear from the context). We use $\cS^{d-1}$ to denote the set of points $v \in \R^d$ with $\|v\|_2=1$. We use  $\top$ for the transpose of matrices and vectors.
For a subspace $\mathcal{V}$ of $\R^d$ of dimension $m$, we denote by $\P_{\cV} \in \R^{d \times d}$ the orthogonal projection matrix of $\cV$. 
 That is, if the subspace $\cH$ is spanned by the columns of the matrix $\bA$, then $ \P_{\cH}:=\bA(\bA^\top \bA)^{-1} \bA^\top$. By slightly overloading notation, if $\bA$ is a matrix, we will also use $\P_\bA$ to denote the orthogonal projection matrix for the subspace spanned by the columns of $\bA$.
We say that a symmetric $d\times d$ matrix $\bA$ is PSD (positive semidefinite) and write $\bA\succcurlyeq 0$ if for all $x\in \mathbb{R}^d$ it holds $x^\top \bA x\ge 0$. We use $\|\bA\|_{\op}$ for the operator (or spectral) norm of the matrix $\bA$. We use $\tr(\bA)$ to denote the \emph{trace} of the matrix $\bA$ and  $\langle \bA ,\bB \rangle = \tr(\bA \bB^\top)$ to denote the \emph{Frobenius inner product} between matrices $\bA$ and $\bB$. For a PSD matrix $\bM$ and a vector $x$,  $\|x\|_\bM := \sqrt{x^\top \bM x}$ denotes the \emph{Mahalanobis norm} of $x$ with respect to $\bM$.

We write $x\sim D$ for a random variable $x$ following the distribution $D$ and use $\E[x]$ for its expectation. We use $\cN(\mu,\vec \Sigma)$ to denote the Gaussian distribution with mean $\mu$ and covariance matrix $\vec \Sigma$. We write $\pr(\cE)$ for the probability of an event $\cE$. We write $\1_{\cE}$ for the indicator function of the event $\cE$.

We use $a\lesssim b$ to denote that there exists an absolute universal constant $C>0$ (independent of the variables or parameters on which $a$ and $b$ depend) such that $a\le Cb$.
Sometime, we shall abuse the notation and use $a = O(b)$ to denote the same to save space.

\paragraph{Projection matrices and convex relaxations.} We use $\cV_k$ to denote the set of all rank-$k$ projection matrices in $\R^{d \times d}$.
Recall that for any $\bV \in \cV_k$, $\bV$ is symmetric, PSD, and idempotent.  
We use $\cM_k$ to denote the set of convex relaxation of $\cV_k$, i.e.,
\begin{align}
\label{eq:M_k_definition}
    \cM_k:= \{ \bM \in \R^{d \times d} : \bM \succeq 0,\; \bM \preceq \bI,\; \trace(\bM) = k\}.
\end{align}

\paragraph{Empirical mean and second moment matrices.} For a  $S \subset \R^d$,
we use the following notation for the sample mean, sample covariance, and the centered second moment matrix with respect to $\mu$ (which shall be clear from context), respectively:
\begin{align}\label{eq:shortcut}
    \mu_S := \frac{1}{|S|}\sum_{x \in S}x,\quad  \vec \Sigma_S := \frac{1}{|S|}\sum_{x \in S}(x-\mu_S)(x-\mu_S)^\top, \quad \overline {\vec \Sigma}_S := \frac{1}{|S|}\sum_{x \in S}(x-\mu)(x-\mu)^\top \;.
\end{align}

\subsection{Generalized Rank-$k$ Stability}

As outlined in \Cref{sec:intro}, our algorithm for mean estimation relies on the exact same \emph{stability condition} developed in prior work. However, for our distribution learning result, our algorithm is a multi-dimensional generalization of the standard filtering, which requires us to consider an appropriate generalization of the stability condition, presented in \Cref{def:condition}. For $k=1$ this definition reduces to the standard stability condition.

\GENERALIZEDSTABILITY*

\begin{remark}
    Using convexity arguments, it can be seen that we can replace the condition   ``for every $\bV \in \cV_k$'' with ``for every $\bM \in \cM_k$'' (cf. \eqref{eq:M_k_definition}) in the second condition of \Cref{def:condition}. 
\end{remark}

\subsubsection{Equivalent Definitions of Generalized Stability}

We will often need to use basic properties of the stability condition that follow directly from its definition. These properties are presented in \Cref{lem:equivalence} as equivalent ways to define the stability condition. These equivalences have been shown in the literature for the special case of $k=1$ (see, e.g., Claim 4.1 in \cite{DiaKP20} and Lemma 3.1 in \cite{DiaKan22-book}), but the proof readily extends to general $k$.\looseness=-1

\begin{lemma}\label{lem:equivalence}
    \Cref{def:condition,def:condition_2,def:condition_3} are all equivalent to each other, up to an absolute constant factor in front of the parameter $\delta$.
\end{lemma}

\begin{definition}[Generalized Stability; Alternative Definition I]\label{def:condition_2}
    Let $\eps\in(0,1/2)$ and $\delta \in [\eps,\infty)$.    Let $S$ be a set of points in $\R^d$ and $\mu$ be a vector. We say that $S$ satisfies the $(\eps,\delta,k)$-generalized-stability condition with respect to $\mu \in \R^d$ if the following holds for every $\bM \in \cM_k$:
       (i) $\left\|  \mu_S - \mu \right\|_{2} \leq  \delta$,
        (ii) $  \left\langle \bM, \overline{\vec \Sigma}_S - \bI \right\rangle  \leq  \delta^2/\eps$, and
        (iii) For all $S' \subset S$ with $|S'| \geq (1-\eps)|S|$ it holds $\left\langle  \bM, \overline{\vec \Sigma}_{S'} - \bI \right\rangle \geq - \delta^2/\eps$.
\end{definition}

\begin{definition}[Generalized Stability; Alternative Definition II]\label{def:condition_3}
    Let $\eps\in(0,1/2)$ and $\delta \in [\eps,\infty)$.    Let $S$ be a set of points in $\R^d$ and $\mu$ be a vector. 
    We say that $S$ satisfies the $(\eps,\delta,k)$-generalized-stability condition with respect to $\mu\in \R^d$ if for every $\bM \in \cM_k$, the set $S$ satisfies the first two conditions of \Cref{def:condition_2}  and it also satisfies the following condition:  
For all $T \subset S$ with $|T| \leq \eps|S|$ it holds that $\frac{1}{|S|}\sum_{x\in T}\|x-\mu\|_{\bM}^2 =   \tfrac{|T|}{|S|} \langle  \bM,  \overline{\vec \Sigma }_T 
\rangle  \leq \delta^2/\eps$.
\end{definition}

\subsection{Consequences of (Generalized) Stability}

\noindent The next result gives a bound on the average of $\|x - \mu\|_{\bM}$ over a small subset of a stable set.

\begin{lemma}
\label{lem:sum-of-absolute-values-rank-k}
Let $S$ be a finite multiset of $n$ points  satisfying the $(\eps,\delta,k)$-generalized-stability condition with respect to $\mu \in \R^d$. Then 
       $\max_{\bM\in \cM_k} \max_{T \subset S:|T| \leq \epsilon |S|} \frac{1}{|S|} \sum_{x \in T}
       \|x - \mu\|_{\bM}
       \lesssim   \delta$.
\end{lemma}
\begin{proof}
Using Cauchy-Schwarz inequality and the last condition in \Cref{def:condition_3}, we obtain
\begin{align*}
     \max_{\bM\in \cM_k} \max_{\substack{T \subset S \\ |T| \leq \epsilon n}} \frac{1}{|S|} \sum_{x \in T} \|x\|_{\bM}
&\leq
\max_{\bM\in \cM_k} \max_{\substack{T \subset S \\ |T| \leq \epsilon n}} \frac{|T|}{|S|} \sqrt{\frac{1}{|T|}\sum_{x \in T} \|x\|_{\bM}^2}\\
&=
\max_{\bM\in \cM_k} \max_{\substack{T \subset S \\ |T| \leq \epsilon n}} \sqrt{\frac{|T|}{|S|}} \sqrt{\frac{1}{|S|}\sum_{x \in T} \|x\|_{\bM}^2}\lesssim \sqrt{\eps} \sqrt{\frac{\delta^2}{\eps}} \leq \delta.
\end{align*}

\end{proof}

\begin{lemma}
\label{lem:stability-bigger-param}
    Let $S$ be an $(\eps,\delta)$-stable set with respect to $\mu$.
    Then 
    $S$ is also $(\eps,\delta',k)$-generalized stable with respect to $\mu$ with $\delta' \lesssim \delta \sqrt{k}$.
\end{lemma}
\begin{proof}
   Since the mean condition in the definition of generalized stability is the same as the one in the plain stability, this condition holds trivially. 
   For the covariance condition, let $\bM = \sum_{i=1}^d \lambda_i v_i v_i^\top$ be the spectral decomposition of $\bM$. Then, for a subset $S' \subseteq S$ with $|S'| \geq (1-\eps) |S|$ we have that
   \begin{align*}
       \left| \left\langle \bM, \frac{1}{|S'| } \sum_{x \in S'}xx^\top - \bI \right\rangle \right|
    \leq \sum_{i=1}^d \lambda_i \left| \left\langle v_i v_i^\top, \frac{1}{|S'| } \sum_{x \in S'}xx^\top - \bI \right\rangle \right|
    \leq \sum_{i=1}^d \lambda_i \delta^2/\eps = \tr(\bM)\delta^2/\eps = k\delta^2/\eps \;,
   \end{align*}
   where the second inequality uses the $(\eps,\delta)$-stability of $S$.
\end{proof}
The next result shows that all large subsets of a stable set are stable, and the contamination parameter $\eps$ is ``robust'' to constant prefactors. 
\begin{lemma}
\label{lem:eps-stability-to-2eps}
Let $S$ be $(\eps,\delta,k)$-generalized stable with respect to $\mu$.
Let $r \geq 1$ be such that $r \eps \leq 1/2$.
\begin{enumerate}
    \item    Then $S$ is also $(r\eps, \delta',k)$-generalized stable with respect to $\mu$ for $\delta' \lesssim \delta\sqrt{r}$.
    \item Any subset $S' \subset S$ such that $|S'| \geq (1- r \eps)|S|$,  is also $(\eps,\delta',k)$-generalized-stable with respect to $\mu$ with $\delta' \lesssim \delta\sqrt{r}$. 
\end{enumerate}
\end{lemma}
\begin{proof}
 We start with the first claim, which we show using \Cref{def:condition_3} for the definition of generalized stability (and \Cref{lem:equivalence}, stating that all definitions are equivalent to each other up to an absolute constant in front of the $\delta$). The first two conditions in \Cref{def:condition_3} (about the mean and second moment over all the points in $S$) hold trivially by the  $(\eps, \delta,k)$-generalized stability of $S$. It remains to show the last condition, that for every set $T \subseteq S$ with $|T| \leq r \eps |S|$ it holds $\frac{1}{|S|}\sum_{x \in T} \|x - \mu \|_\bM^2 \leq \delta'^2/\eps$. This can be seen by splitting $T$ into at most $r$ disjoint sets of size at most $\eps |S|$ each, and apply the corresponding condition from the $(\eps, \delta,k)$-generalized stability of $S$. That is, write $T = T_1 \cup \cdots \cup T_{r'}$ where $T_j$ are disjoint, $|T_j| \leq \eps |S|$ and $r'\leq r$. Then
 \begin{align*}
     \frac{1}{|S|}\sum_{x \in T} \|x - \mu \|_\bM^2 = \sum_{j=1}^{r'} \frac{1}{|S|}\sum_{x \in T_j} \|x - \mu \|_\bM^2 \leq r' \delta^2/\eps \;.
 \end{align*}

We move to the second claim.
Using the first claim, we have that $S$ is $((r+1)\eps, \delta',k)$-generalized stable with $\delta' \lesssim \sqrt{r}\delta$ (since $r \geq 1$). It remains to check that the two conditions from \Cref{def:condition} hold for every subset $S''$ of size $|S''| \geq (1-\eps)|S'|$. Since $(1-\eps)|S'| \geq (1-\eps)(1-r\eps)|S| \geq (1- (r+1)\eps))|S|$, the desired conditions follow by the $((r+1)\eps, \delta',k)$-generalized stability of $S$.
\end{proof}

Finally, the next result shows that all large subsets of a stable set are close in the sliced-Wasserstein metrics (\Cref{def:sliced-Wasserstein}).

\begin{lemma}\label{w1-w2-subsets-prelims}
    Let $S_0$ be a set satisfying $(\eps,\delta,k)$-generalized stability, and $S_0'$ be a subset of $S_0$ with $|S_0'| \geq (1-\eps)|S_0|$ for $\eps \leq 1/2$.
    Then, $W_{1,k}(S_0,S_0') \lesssim \eps \sqrt{k} + \delta$ and $W_{2,k}(S_0,S_0') \lesssim \sqrt{\eps k} + \delta/\sqrt{\eps}$.
\end{lemma}

\begin{proof}
Let us use the notation $S_0 = \{ x_1,\ldots,x_n\}$, for our set satisfying $(\eps,\delta,k)$-generalized stability with respect to $\mu$. Let us use $\mu=0$ thought the proof without loss of generality.
 Define $\cJ\subset [n]$ for the set of indices corresponding to the points in $S_0'$ (and $\cJ^\complement = [n] \setminus \cJ$ for the rest of the points), and denote $m := |\cJ| = |S_0'|$.
    Recall the definition of sliced-Wasserstein distance from \Cref{def:sliced-Wasserstein} for $p \in \{1,2\}$:
    \begin{align*}
        W_{p,k}(S_0,S_0') = \sup_{\bV \in \cV_{k}} \inf_{\pi \in \Pi(S_0,S_0')} \E_{(x,x') \sim \pi} \left[ \|\bV(x-x')  \|_2^p \right]^{1/p}\;.
    \end{align*}
    
    \noindent We will use the following coupling $\pi$ on $(x,x')$: 
    \begin{itemize}
        \item First, $x'= x_i$ for an index $i$ chosen uniformly at random from $\cJ$.
        \item Then, conditioned on $x'= x_i$, with probability $m/n$, $x$ is set to be $x_i$ and with probability $1-m/n$, $x$ is chosen to be $x_i$ for an index chosen uniformly at random from $\cJ^\complement$.
    \end{itemize}
 It can be checked that this is a valid coupling: The marginal of $x'$ is the uniform distribution on $S_0'$ (by definition), and the marginal of $x$ is uniform on $S_0$ since for every $i \in \cJ$ we have $\P[x = x_i] = \frac{1}{m}\frac{m}{n}=1/n$ and for every $i \in \cJ^\complement$ we have $\P[x = x_i] = \sum_{i \in [m]} \frac{1}{m} (1-m/n)\frac{1}{n-m}=1/n$. 
 We can thus bound $W_{p,k}(S_0,S_0')$ as follows:

    \begin{align}
        W_{p,k'}(S_0,S_0')^p &\leq \sup_{\bV \in \cV_{k}} \E_{(x,x') \sim \pi}\left[\|\bV(x-x')\|_2^p\right] \notag\\
        &\lesssim \sup_{\bV \in \cV_{k}} \frac{1}{m} \sum_{i \in \cJ}\E_{(x,x') \sim \pi}\left[\|\bV(x-x')\|_2^p \,|\, x' =  x_i\right] \notag\\
        &= \sup_{\bV \in \cV_{k}} \frac{1}{m} \sum_{i \in \cJ}\left( \frac{m}{n}\|\bV (x_i-x_i)\|_2^p + \frac{n-m}{n}\frac{1}{n-m} \sum_{j \in \cJ^\complement}\|\bV(x_j - x_i)\|_2^p \right) \notag \\
        &= \sup_{\bV \in \cV_{k}} \frac{1}{mn} \sum_{i \in \cJ,j\in \cJ^\complement} \|\bV(x_j -   x_i)\|_2^p \;. \label{eq:w_1}
    \end{align}

\paragraph{Controlling $W_{1,k}$.} 
    We first consider the easy case of $p=1$, for which we can use the triangle inequality to obtain the following:   
    \begin{align*}
        \frac{1}{mn} \sum_{i \in \cJ,j\in \cJ^\complement} \|\bV(x_j -  x_i)\|_2
        &\leq \frac{1}{mn}\sum_{i \in \cJ,j\in \cJ^\complement} \left(\|\bV x_j\|_2 + \|\bV x_i\|_2\right) \\
        &\lesssim \frac{1}{n} \sum_{j \in \cJ^\complement} \|\bV x_j\|_2 +  \eps\cdot\frac{1}{m} \sum_{i \in \cJ} \|\bV  x_i\|_2  \\
        &\lesssim \delta + \eps\cdot\frac{1}{n} \sum_{i \in [n]} \|\bV  x_i\|_2 \;,
        \numberthis\label{eq:w-1-decomposition}
    \end{align*}
    where the bound on the first term follows by  \Cref{lem:sum-of-absolute-values-rank-k}, and the bound on the second term uses that $n \lesssim m$.
    We now use the $(\eps,\delta,k)$-generalized-stability of $S_0$ and Cauchy-Schwarz inequality to obtain the following: 
    \begin{align*}
     \frac{1}{n} \sum_{i \in [n]} \|\bV x_i\|_2
        \leq  \sqrt{ \frac{1}{n} \sum_{i \in [n]} \|\bV x_i\|_2^2} 
        =  \sqrt{ \langle \bV, \vec\Sigma_{S_0}\rangle} 
        \lesssim  \sqrt{k + \delta^2/\eps} \lesssim \sqrt{k} + \delta/\sqrt{\eps} \;.
    \end{align*}
This concludes an upper bound on $W_{1,k}(S_0,S_0')$ of the order $\delta + \eps \sqrt{k} + \delta \sqrt{\eps}$.
    \paragraph{Controlling $W_{2,k}$.}
    We now turn our attention to $W_{2,k}(S_0,S_0')$ using $p=2$ in \Cref{eq:w_1}. Using the inequality $(a+b)^2\leq 2a^2 + 2 b^2$ to analyze the cross term, we obtain:
    \begin{align*}
        \frac{1}{nm} \sum_{i \in \cJ}  \sum_{j \in \cJ^\complement} \|\bV( x_j -  x_i)\|_2^2 
        &\lesssim  \frac{\eps}{m} \sum_{i \in \cJ} \|\bV x_i\|_2^2 + \frac{1}{n} \sum_{j \in \cJ^\complement}\|\bV x_j\|_2^2\\
        &\lesssim  \eps\langle \bV,\vec\Sigma_{S_0}\rangle    +  \frac{\delta^2}{\eps}
        \lesssim \eps \left( k + \frac{\delta^2}{\eps} \right)   +  \frac{\delta^2}{\eps}
        \lesssim   \eps k + \frac{\delta^2}{\eps}.
    \end{align*}
    where the bounds in the last line follow by the generalized stability of the original points $S_0$.
    Thus, we conclude that $W_{2,k}(S_0,S')^2 \lesssim  \eps k + \delta^2/\eps$.

\end{proof}

\section{Averages of Low-rank Projections and Their Convex Relaxations}
\label{sec:rounding}
In this section, we derive the crucial structural property of the local perturbation and its (generalized) projections.
To elaborate,  let $\{\Delta_i\}_{i=1}^n$ be the local perturbations satisfying \Cref{def:strongWass-k} (the sliced Wasserstein distance).
Towards our ultimate goal of establishing stability of the locally perturbed data, we need to argue that the $\{\|\Delta_i\|_{\bM}\}_{i \in [n]}$ behaves \emph{nicely} for any $\bM \in \cM_k$.
While for any projection $\bV \in \cV_k$, the desired \emph{niceness} of $\{\|\Delta_i\|_{\bV}\}_{i \in [n]}$ follows directly from \Cref{def:strongWass-k},
our proof arguments necessitate understanding the behavior of the generalized projections for $\bM \in \cM_k$.
The reason for considering these general matrices stems from a convex relaxation needed to employ the
min-max duality theorem in the proof of \Cref{prop:subset-covariance-general-k}.

\begin{proposition}[Bound on Average Projections]\label{lem:average_roots}
    Let $y_1, y_2, \dots, y_n$ be vectors in $\R^d$.
    Then the following holds:
        $\sup_{\bM \in \cM_k} \frac{1}{n} \sum_{i=1}^n \|y_i\|_{\bM} \lesssim \sup_{\bV \in \cV_k}\frac{1}{n}\sum_{i=1}^n \|y_i\|_\bV$.
\end{proposition}
We prove the above result using a Gaussian rounding scheme, inspired by \cite{DepLec22}.

\begin{proof}
    Let $\rho := \sup_{\bV \in \cV_k}\frac{1}{n}\sum_{i=1}^n \|\bV y_i\|_2$.
    Suppose that there exists a $\bM \in \cM_k$ with the property  $\frac{1}{n}\sum_{i=1}^n \sqrt{y_i^\top \bM y_i} \geq C' \cdot\rho$ for a sufficiently large constant $C'$.
    We will show that this leads to a contradiction.
    
        For a $r>0$, let $\cB_{r} := \{\bB \in \R^{d \times d} \succeq 0: \|\bB\|_\op \leq r, \; \bB \text{ is rank-$k$}\}$ be the set of rank-$k$ PSD matrices with bounded operator norm.
        Let $g_1,\dots,g_k$ be i.i.d.\ samples from $\cN(0,\bM)$
        and define $\bB := \frac{1}{k} \sum_{i=1}^k g_ig_i^\top$ to be the empirical second moment matrix, which is an unbiased estimate of $\bM$.
    We define the random variable
    \begin{align*}
        Z = \frac{1}{n} \sum_{i=1}^n  \|y_i\|_{\bB} \1_{\bB \in \cB_{r}} \,.
    \end{align*}
    On the one hand, we have that
    \begin{align}
        Z \leq \frac{1}{n} \sum_{i=1}^n \|\bB^{1/2}\|_\op \| y_i\|_{\bP_{\bB}} \1_{\bB \in \cB_{r}} \leq \sqrt{r} \sup_{\bV \in \cV_k} \frac{1}{n} \sum_{i=1}^n  \| y_i\|_{\bV} \leq \sqrt{r}\rho,
    \end{align}
    where we use that $\bB$ is a rank-$k$ matrix if $\bB \in \cB_{r}$.
    The above implies that $\E[Z] \leq \sqrt{r} \rho$.
    We shall now show contradiction by deriving a lower bound on $\E[Z]$.
    
    Towards that goal, we use the following decomposition to handle the indicator event $\bB \in \cB_r$: 
    \begin{align}
        \E[Z] = \frac{1}{n }\sum_{i=1}^n  \E\left[ \| y_i\|_{\bB} \right]
        -\frac{1}{n }\sum_{i=1}^n  \E\left[ \| y_i\|_{\bB}  \1_{\bB \not\in \cB_r} \right] \label{eq:twoparts}
    \end{align}
    where the expectation is taken with respect to the random matrix $\bB$.

    We first obtain a lower bound on the first term above.
    Consider the random variables $R_i := \|y_i\|_{\bB}^2 =        \frac{1}{k}\sum_{j=1}^k (g_j^\top y_i)^2 $, which is a degree-two polynomial of the Gaussian samples.
    Then $\E[R_i] = \|y_i\|_{\bM}^2$. 
    To obtain a lower bound on $\E[\sqrt{R_i}]$, we shall prove an upper bound on $\E[R_i^2]$.
    Using $\E[G^4] \leq 3 (\E[G^2])^2$ for a Gaussian variable $G$, we obtain
            $\E[R_i^2] = \frac{1}{k^2}\sum_{j=1}^k \E[ (g_j^\top y_i)^4 ] + \frac{k(k-1)}{k^2} (\E[R_i])^2 
                    < \frac{k+3}{k}  (\E[R_i])^2 \leq 4 (\E[R_i])^2$. 
    We shall use this upper bound in the following inequality: $\E[|X|]^{2/3} \E[|X|^4]^{1/3} \geq \E[|X|^2]$ which holds for any real-valued random variable $X$ with finite fourth moments. 
    Applying it to $\sqrt{R_i}$, we obtain 
    $\E[\sqrt{R}_i] \geq \frac{\E[R_i]^{3/2}}{\E[R_i^2]^{1/2}} \geq \frac{\E[R_i]^{3/2}}{2\E[R_i]} = \frac{1}{2} \sqrt{\E[R_i]}$,
    where the middle step uses the aforementioned upper bound on $\E[R_i^2]$.
    Combining everything, we have shown the following lower bound on the first term:
    \begin{align*}
        \frac{1}{n }\sum_{i=1}^n  \E\left[ \|y_i\|_{\bB} \right]
        \geq \frac{1}{2} \frac{1}{n }\sum_{i=1}^n  \sqrt{\E\left[ \|y_i\|_\bB^2 \right]}
        = \frac{1}{2}  \frac{1}{n }\sum_{i=1}^n \sqrt{ \|y_i\|_\bM^2}.
    \end{align*}
    where the second step uses that $\E[\bB] \ = \bM$. 
    We now show that the second term in \eqref{eq:twoparts} can be ignored as follows:
    \begin{align*}
        \frac{1}{n }\sum_{i=1}^n  \E\left[ \| y_i\|_\bB  \1_{\bB \not \in \cB_r} \right] 
        &\leq \frac{1}{n }\sum_{i=1}^n  \sqrt{\E\left[ \| y_i\|_\bB^2  \right] } \sqrt{\pr[\bB \not\in \cB_r]} \tag{Cauchy-Schwarz} \\
        &\leq \frac{1}{n }\sum_{i=1}^n  \|y_i\|_{\bM} \sqrt{\frac{4C}{r}}, 
    \end{align*}
    where the last inequality follows by concentration of covariance of Gaussians, which we establish next.
        First, we observe that $\bM$
        must have rank at least $\trace(\bM)/\|\bM\|_\op \geq k$. 
        Since $g_1,\dots,g_k$ are sampled i.i.d.\ from $\cN(0,\bM)$ with $\rank(\bM) \geq k$,  
         the matrix $\bB := \sum_{i}g_ig_i^\top$ has rank exactly $k$
         has rank exactly $k$ with probability $1$; This is because the Lebesgue measure of a rank-deficient subspace is zero.
         Thus, it remains to show that $\bB$ has operator norm at most $r$ with probability at least $1 - 1/32$.
        Observe that $\bB$ is the second moment matrix of $k$ independent Gaussians whose covariance matrix has trace equal to $k$.
        Applying the Gaussian covariance concentration  results to this setting (see, e.g.,  \cite[Theorem 4]{KoltchinskiiLounici17} or  \cite[Theorem 5.1]{VanHandel2017}), we obtain that for an absolute constant $C$:
        \begin{align*}
        \E[\|\bB - \bM\|_\op] \leq C \|\bM\|_\op \left( \sqrt{\frac{1}{\|\bM\|_\op}} + \frac{1}{\|\bM\|_\op} \right) \leq C (1 + \sqrt{\|\bM\|_\op}) \leq 2C. 
        \end{align*}
        Since for any $r \geq 2$, $\bB \not\in \cB_{r}$ implies that
        $\|\bB-\bM\|_\op \geq r - 1 \geq r/2$,
        applying the Markov inequality, we obtain the desired inequality 
        $\pr(\bB \not \in \cB_{r/2}) \leq \frac{4C}{r}$.

    Putting everything together and taking  $r= 64C$, we obtain the following:
    \begin{align}
        \E[Z] \geq \left(\frac{1}{2} - \sqrt{\frac{4C}{r}} \right)\frac{1}{n }\sum_{i=1}^n  \|y_i\|_{\bM} = \frac{1}{4} \frac{1}{n }\sum_{i=1}^n  \|y_i\|_{\bM} \geq C'\rho/4, 
    \end{align}
     where we used the assumption $\sum_{i=1}^n  \|y_i\|_{\bM} > C'\rho$.
    If $C'/4 > \sqrt{r} = \sqrt{128C}$, this contradicts the upper bound $\E[Z] \leq \sqrt{r} \rho$ established earlier.
\end{proof}

\section{Stability Is Preserved Under Local Perturbations}
\label{sec:stability-under-local}

In this section, we present the main technical results behind \Cref{thm:mean-estimation,thm:main_distr_learning}.
First, we restate \Cref{prop:subset-covariance-intro}, which establishes the existence of a large subset of local perturbations whose second moment is bounded appropriately.
\begin{proposition}
\label{prop:subset-covariance-general-k}
    Let  $z_1,\dots,z_n$ be vectors in  $\R^d$  satisfying       $ \max_{\bV \in \cV_k} \frac{1}{n} \sum_{i \in [n]}\|z_i\|_{\bV} \leq \rho$.
    Then for every $\eps \in (0,1)$ there exists a subset $\cI \subseteq [n]$ such that (i) $|\cI| \geq (1-\eps)n$ and 
    (ii)
    \begin{align}
    \max_{\bV \in \cV_k} \frac{1}{|\cI|}\sum_{i \in \cI} \| z_i\|_{\bV}^2 = \max_{\bM \in \cM_k} \frac{1}{|\cI|}\sum_{i \in \cI} \| z_i\|_{\bM}^2 \lesssim \rho^2/\eps\,.    
    \end{align}
    
\end{proposition}
The proof of this result was provided in  \Cref{sec:techniques}. 
The first equality is simply because $\cM_k$ is the convex hull of $\cV_k$ and $\bV \bV^\top = \bV$ for all $\bV \in \cV_k$. 

We shall use the above result and \Cref{lem:average_roots}
in combination with the following result, stating that as long as local perturbations have bounded second moment, the stability parameter degrades in a dimension-independent manner.
\begin{theorem}
\label{lem:existence-of-stability-and-W_1-and-W-2-new}
    Let $S_0' = \{x_1,\dots,x_n\}$ be an $(\eps,\delta,k)$-generalized stable set with respect to $\mu\in \R^d$.
    Let $\Delta_1,\dots,\Delta_n \in \R^d$ be vectors satisfying
    \begin{align}
    \label{eq:assumption-on-Delta-new}
        \max_{\bM \in \cM_k} \frac{1}{n}\sum_{i \in [n]} \|\Delta_i\|_\bM \lesssim \rho \qquad \text{ and } \qquad 
        \max_{\bM \in \cM_k} \frac{1}{n}\sum_{i \in [n]}  \|\Delta_i\|_\bM^2 \lesssim \frac{\rho^2}{\eps} \;. 
    \end{align}
    Define $\widetilde{x}_i := x_i + \Delta_i$ for all $i \in [n]$ and define the set $S'$ to be $\{\widetilde{x}_i\}_{i \in [n]}$.
    Then the following hold:
  \begin{itemizec}      
        \item \label{it:det-condition-mean-new} $S'$
        satisfies
        $(\eps,\widetilde{\delta},k)$-generalized stability
        with respect to $\mu$ (\Cref{def:condition}) for $\widetilde \delta \lesssim \delta + \rho$.
        \item \label{it:det-condition-W_1k-new}
        For all large subsets $S''\subset S'$ with $|S''| \geq (1-\eps)|S'|$:
        $W_{1,k} (S_0',S'') \lesssim \rho + \eps \sqrt{k} + \delta $.         

        \item 
            For all large subsets $S''\subset S'$ with $|S''| \geq (1-\eps)|S'|$:
        \label{it:det-condition-W_2k-new} $W_{2,k} (S_0', S'') \lesssim 
        \sqrt{\eps k} + \frac{\delta}{\sqrt{\eps}} + \frac{\rho}{\sqrt{\eps}}$.
  \end{itemizec}
\end{theorem}

 We break down the proof into two separate statements, showing the stability and closeness in Wasserstein distance individually:
\begin{restatable}{lemma}{LemExistenceStabilityNew}    
\label{lem:existence-of-stability-new}
Consider the setting in \Cref{lem:existence-of-stability-and-W_1-and-W-2-new}.
 Then $S'$ satisfies $(\eps,\widetilde{\delta},k)$-generalized-stability with respect to $\mu$  for $\widetilde \delta \lesssim \delta + \rho$.
\end{restatable}
\begin{restatable}{lemma}{LemExistenceWassNew}    
\label{lem:existence-of-W_1-and-W-2-new}
Consider the setting in \Cref{lem:existence-of-stability-and-W_1-and-W-2-new}.
 Then for all $|S''| \geq (1-\eps)|S'|$, we have $W_{1,k} (S_0',S'') \lesssim \rho + \eps \sqrt{k} +  \delta$ and $W_{2,k} (S_0', S'') \lesssim \sqrt{\eps k} + \delta/\sqrt{\eps} + \rho/\sqrt{\eps}$.
\end{restatable}

 \Cref{lem:existence-of-stability-and-W_1-and-W-2-new} follows trivially from the above two results.
We  provide the proofs of \Cref{lem:existence-of-stability-new,lem:existence-of-W_1-and-W-2-new} in the next two sections.

\subsection{Proof of \Cref{lem:existence-of-stability-new}}

\begin{proof}[Proof of \Cref{lem:existence-of-stability-new}]
We use $\mu=0$ throughout this proof without loss of generality.
To establish the desired  stability result, \Cref{lem:equivalence} implies that it is equivalent (up to a constant factor in $\widetilde{\delta}$) to establishing the following conditions (simultaneously):
\begin{enumerate}
  \item (Mean) $\left\|\mu_{S'}  \right\|_{2} \leq \widetilde{\delta}$.
  \item (Upper bound on covariance) For all $\bM \in \cM_k$:  $\left\langle \bM, \overline{\vec \Sigma}_{S'} - \bI \right\rangle  \leq \widetilde{\delta}^2/\eps$.
  \item (Lower bound on all large subsets) For all $S''\subset S'$ with $|S''| \geq (1-\eps)|S'|$ and all $\bM \in \cM_k$:
  $\left\langle \bM, \overline{\vec \Sigma}_{S'} - \bI \right\rangle  \geq - \widetilde{\delta}^2/\eps$.
\end{enumerate}

We will establish these three conditions separately.

\paragraph{Mean condition.}
We start with the mean condition, which follows directly by triangle inequality, the stability of the original points $\{x_i\}_{i \in [n]}$, and the assumption \eqref{eq:assumption-on-Delta-new}:
    \begin{align*}
        \left\| \frac{1}{n}\sum_{i \in [n]}  \widetilde{x}_i     \right\|_2
        &\leq \left\| \frac{1}{n}\sum_{i \in [n]} x_i   \right\|_2
         + \left\| \frac{1}{n}\sum_{i \in [n]} \Delta_i  \right\|_2 
         \lesssim \delta + \sup_{v \in \cS^{d-1}} \sum_{i=1}^n |v_i^\top \Delta_i|         \lesssim \delta + \rho\;.
    \end{align*}
    
    \paragraph{Upper bound on covariance.}
 Using the decomposition $\widetilde{x}_i = x_i + \Delta_i$ yields the following:
    \begin{align}
        \max_{\bM \in \cM_k} \left\langle \bM, \overline{\vec \Sigma}_{S} - \bI \right\rangle &=  \frac{1}{n}    \sum_{i \in [n]}  \|\widetilde{x}_i\|_{\bM}^2 - \trace(\bM) \notag\\       
        &=
 \frac{1}{n}\sum_{i \in [n]}\left( \left(  \|x_i\|_{\bM}^2 - \trace(\bM) \right) +  \|\Delta_i\|_{\bM}^2 \label{eq:4terms1-new} +  2x_i ^\top \bM   \Delta_i \right) \;.
    \end{align}
    We now bound each of the terms above separately. 
    \begin{itemize}
      \item     The first  expression  \Cref{eq:4terms1-new} can be handled by the generalized stability of the original points in $S_0'$. Since $\bM \in \cM_k$, and $S_0$ is $(\eps,\delta,k)$-stable, we obtain that $ \frac{1}{n}\sum_{i \in [n]}  \|x_i\|_{\bM}^2 - \tr(\bM) \leq \delta^2/\eps$.

\item     The next term in \eqref{eq:4terms1-new} is at most $\rho^2/\eps$ by \Cref{eq:assumption-on-Delta-new}.

\item     We finally bound the cross term in \eqref{eq:4terms1-new}. 
For any fixed $\bM \in \cM$, define $\cI_\bM$ to be the set of indices in $[n]$ such that $\|x_i\|_{\bM}$ is bigger than $C \delta/\eps$.
For a large enough constant $C$, the last condition in \Cref{def:condition_3} implies that $|\cI_\bM| \leq \eps$.
Combining this upper bound on the cardinality of $\cI_\bM$ with the last condition in \Cref{def:condition_3} we have that $\tfrac{1}{n}\sum_{i \in \cI_\bM}\|x_i\|_\bM^2 \lesssim \frac{\delta^2}{\eps}$.
Using the Cauchy-Schwarz inequality, we obtain that for any $\bM \in \cM_k$:
    \begin{align*}
        \frac{1}{n}\sum_{i \in [n]}  x_i ^\top \bM  \Delta_i
        &\leq \frac{1}{n}\sum_{i \in [n]}   \|\Delta_i\|_{\bM} \|x_i\|_{\bM} \\
        &\leq \frac{1}{n}\sum_{i \in [n]}   \|\Delta_i\|_{\bM}\|x_i\|_\bM \1_{i \not \in \cI_\bM} 
         + \frac{1}{n}\sum_{i \in [n]}   \|\Delta_i\|_{\bM}\|x_i\|_\bM\1_{i \in \cI_\bM} \\        
        &\lesssim  \frac{1}{n}\sum_{i \in [n]} \frac{\delta}{\eps}  \|\Delta_i\|_{\bM} + \sqrt{\frac{1}{n}\sum_{i \in [n]}   \|\Delta_i\|_{\bM}^2 } \sqrt{ \frac{1}{n}\sum_{i \in \cI_\bM}   \|x_i\|_{\bM}^2 } \\
        &\lesssim \frac{\delta \rho}{\eps} + \sqrt{\frac{\rho^2}{\eps}} \sqrt{\frac{\delta^2}{\eps}}
         \,\, \lesssim \frac{\delta \rho}{\eps} \;.
     \label{eq:cross-term-M-new} \numberthis
     \end{align*}
    where the last step uses \Cref{eq:assumption-on-Delta-new} and $\tfrac{1}{n}\sum_{i \in \cI_\bM}\|x_i\|_\bM^2 \lesssim \frac{\delta^2}{\eps}$ which we have shown earlier.

    \end{itemize}

    Combining everything, we have shown that for all matrices $\bM \in \cM_k$, the following bound holds:
    $\langle \bM, \overline{\vec \Sigma}_{S'} - \bI \rangle \lesssim  \frac{ \delta^2}{\eps} + \frac{\rho^2}{\eps} + \frac{\delta\rho}{\eps} \leq \frac{\widetilde{\delta}^2}{\eps}$. Therefore $\langle \bM, \overline{\vec \Sigma}_{S'} - \bI \rangle \lesssim \tilde \delta^2/\eps$  for some $\widetilde{\delta} \lesssim \delta + \rho$.

    \paragraph{Lower bound on covariance.}\looseness=-1
    For any $\bM$ and any subset $S'' \subset S'$, the decomposition in \Cref{eq:4terms1-new} yields
    \begin{align*}
\frac{1}{|S''|} \sum_{i: \tilde x_i \in S''} \|\widetilde{x}_i \|_{\bM}^2 - \trace(\bM)
 &= \frac{1}{|S''|} \sum_{i : \tilde x_i \in S''} \left( \left(\|x_i\|_{\bM}^2 - \trace(\bM)\right) + \|\Delta_i\|_{\bM}^2 
        + 2 x_i^\top \bM \Delta_i     \right)\,.
    \end{align*}
    Since $|S''| \geq (1-\eps) |S'| \geq (1-2\eps)|S_0'|$, the first term is at least $ - O(\delta^2/\eps)$ 
    by the stability of $S_0'$ (the last part in \Cref{def:condition_2}) and \Cref{lem:eps-stability-to-2eps}.
    The second term is non-negative. The third term has absolute value at most $\delta\rho/\eps$ because \Cref{eq:cross-term-M-new} shows that the average value of the absolute value of the cross terms is small.
    Thus, we have shown the following lower bound on the covariance:
    for all subsets $S'' \subset S$ with $|S''| \geq (1-\eps)|S'|$, it holds that 
    $\left\langle \bM, \overline{\vec \Sigma}_{S''} - \bI \right\rangle \geq - \widetilde{\delta}^2/\eps$ for some $\widetilde{\delta} \lesssim \delta + \rho$.

\end{proof}

\subsection{Proof of \Cref{lem:existence-of-W_1-and-W-2-new}}
\LemExistenceWassNew*
\begin{proof}[Proof of \Cref{lem:existence-of-W_1-and-W-2-new}]
Using the triangle inequality:
    \begin{align*}
        W_{1,k}(S_0',S'') \leq W_{1,k}(S_0',S') + W_{1,k}(S',S'') \;.
    \end{align*}
    For the first term, note that $S_0'$ and $S'$ have the same cardinality, and for every point $x_i \in S_0'$ we have exactly one point $\tilde x_i \in S'$ with $x_i - \tilde x_i = : \Delta_i$. Thus we can use this simple coupling in \Cref{def:sliced-Wasserstein} to get this upper bound:  
    \begin{align*}
        W_{1,k}(S_0',S') \leq \max_{\bV \in \cV_k} \frac{1}{|S_0'|}\sum_{i: i\in S_0'} \| \Delta_i \|_2 \lesssim \rho
    \end{align*}
    where the last inequality follows by \eqref{eq:assumption-on-Delta-new}.
    For the second term we have $W_{1,k}(S',S'') \lesssim \eps \sqrt{k} + \widetilde{\delta} \leq \eps \sqrt{k} + \rho + \delta$ by \Cref{w1-w2-subsets-prelims}, which is applicable because $S''$ is an $(1-\eps)$-subset of $S'$ and we have already shown in \Cref{lem:existence-of-stability-new} that $S'$ is $(\eps,\tilde \delta,k)$-generalized stable with $\tilde \delta \lesssim \delta + \rho$.

    The bound for $W_{2,k}(S_0',S'')$ is similar. First, $W_{2,k}(S_0',S'') \leq W_{2,k}(S_0',S') + W_{2,k}(S',S'')$. The first term is $O(\rho/\sqrt{\eps})$ by \eqref{eq:assumption-on-Delta-new} and the second term is at most $O(\sqrt{\eps k} + \tilde{\delta}/\sqrt{\eps})= O(\sqrt{\eps k} +  \delta/\sqrt{\eps} + \rho/\sqrt{\eps})  $ by \Cref{w1-w2-subsets-prelims}.
\end{proof}

\section{Mean Estimation: Proof of \Cref{thm:mean-estimation}}
\label{sec:mean-estimation}
In this section, we prove the first algorithmic result of our paper (\Cref{thm:mean-estimation}),
whose proof follows easily from \Cref{lem:average_roots}, \Cref{prop:subset-covariance-general-k} and \Cref{lem:existence-of-stability-and-W_1-and-W-2-new}.

\ThmMeanEstStability*

\begin{proof}[Proof of \Cref{thm:mean-estimation}]
    Let $S_0 = \{x_i\}_{i \in [n]}$ be an $(\eps,\delta)$-stable set with respect to $\mu \in \R^d$. 
    The final set $T$ is constructed by first picking $S \in \cW^\strong(S_0,\rho)$ (cf. \Cref{def:strongWass-intro}) and then $T \in \cO(S,\eps )$ (cf. \Cref{def:outliers}) by the corresponding adversaries. Let $\Delta_i$ denote the perturbabtions of $\cW^\strong$ adversary, i.e., $S = \{\tilde x_i \}_{i \in [n]}$ where $\tilde x_i = x_i + \Delta_i$.
    By \Cref{prop:subset-covariance-general-k} applied with $z_i = \Delta_i$ and $k=1$, there exists a subset $S_0' \subset S_0$ of size at most $(1-\eps)n$ for which $\max_{\bM \in \cM_1} \tfrac{1}{|S_0'|}\sum_{i: x_i \in S_0'} \| \Delta_i\|_{\bM}^2 \lesssim \rho^2/\eps$. Applying \Cref{lem:average_roots} with $k=1$ and $y_i = \Delta_i$ for the set $S_0'$,
    we have that $\sup_{\bM \in \cM_1} \tfrac{1}{|S_0'|}\sum_{i: x_i \in S_0'}   \|\Delta_i\|_{\bM} \lesssim \sup_{\bV \in \cV_k} \tfrac{1}{|S_0'|}\sum_{i: x_i \in S_0'} \|\Delta_i\|_\bV \lesssim \rho$ (where the last inequality follows by the definition of the local perturbations model).
    
    So far, we have established the necessary conditions in \Cref{eq:assumption-on-Delta-new} for applying \Cref{lem:existence-of-stability-and-W_1-and-W-2-new} with $k=1$. The condition in \Cref{lem:existence-of-stability-and-W_1-and-W-2-new} that $S_0'$ is $(\eps,O(\delta),1)$-generalized stable is satisfied because (i) $S_0$ is $(\eps,\delta)$ and (ii) $S_0'$ is a large subset of $S_0$ (\Cref{lem:eps-stability-to-2eps}).
    
    \Cref{lem:existence-of-stability-and-W_1-and-W-2-new}  guarantees that if $S'$ denotes the set $\{x_i + \Delta_i : x_i \in S_0' \}$ (i.e., the points in $S_0'$ after the local perturbations), then  $S'$ is $(\eps,\tilde \delta)$-stable with respect to $\mu$ (cf. \Cref{def:stability1}), for $\tilde \delta \lesssim \rho + \delta$.
    After $T$ is chosen by the second adversary (the one associated with the global outliers), we have that $|T \cap S| \geq (1-\eps)|T|$ which implies that $|T \cap S'| \geq |T \cap S| - |S \setminus S'| \geq (1-2\eps)|T|$.
    This means that $T \in \cO(S',2\eps )$ for an $(2\eps,O(\tilde \delta))$-stable set $S'$ (which follows from $(\eps,\widetilde{\delta})$-stability of $S'$ and \Cref{lem:eps-stability-to-2eps}), and thus any stability-based algorithm outputs a $\hat{\mu}$ such that $\| \hat{\mu} - \mu\|_2 \lesssim \tilde \delta \lesssim \rho + \delta$.
\end{proof}

\section{Distribution Learning Under Global and Local Corruptions}\label{sec:distr_learning}

In this section, we prove \Cref{thm:main_distr_learning}, which yields guarantees for distribution learning in the presence of the combined contamination model.

\MAINTHEOREM*
This result also relies on the structural result of  \Cref{prop:subset-covariance-general-k} and \Cref{lem:existence-of-stability-and-W_1-and-W-2-new}, provided in \Cref{sec:stability-under-local}.
For the distribution learning result, we provide an algorithm which uses a multi-dimensional variant of the standard iterative filtering procedure, given in \Cref{alg:mean_estimation}. Then, leveraging \Cref{lem:existence-of-stability-and-W_1-and-W-2-new}, a certification lemma from \cite{NieGS24} in \Cref{sec:cert}, and a now-standard analysis of the iterative filtering procedure, we prove \Cref{thm:main_distr_learning} in \Cref{sec:distr-learning-proof}.

\label{sec:correctness}

\begin{algorithm}[]
\caption{Distribution learning under global and strong local corruptions}
\label{alg:mean_estimation}
\begin{algorithmic}[1]

\Statex \textbf{Input}: (Multi)-Set of samples $T \subset \R^d$,  and parameters $k' \in \N,\eps \in (0,c),\delta \geq \eps,\rho \geq 0$.
\Statex \textbf{Output}: $\widehat{S}\subset \R^d$ and  $\widehat{\mu} \in \R^d$ such that $\|\widehat{\mu} - \mu\|_2 \lesssim \delta + \rho$ and $W_{1,k} (\widehat{S}, S) \lesssim \delta \sqrt{k'} + \rho $.
\vspace{6pt}

\State Let $C$ be a sufficiently large absolute constant.
\State Define $\tilde \delta = \delta\sqrt{k'} + \rho$.

\vspace{10pt}
\While{true} 
\State Compute $\bM \in \R^{d \times d}$ that maximizes $\langle \bM, \vec \Sigma_T \rangle$ under the constraints $0 \preceq \bM \preceq \bI$, $\tr(\bM)=k'$.
\label{line:optimization-problem}
\If{$\langle \bM, \vec \Sigma_T\rangle \leq k' + C \tilde \delta^2/\eps$}  \label{line:stopping_cond}
    \State Go to  \Cref{line:return}.
\Else
    \State Let $L \subset T$ consisting of the $\eps |T|$ points with the largest score $g(x) = (x - \mu_T)^\top \bM (x - \mu_T)$. 
    \State Define the thresholded scores $\tau(x):= g(x) \1_{x \in L}$ for $x \in T$.\label{line:scores_def}
    \For{each $x \in T$}
        \State Delete $x$ from $T$ with probability $\tau(x)/\max_{x \in T} \tau(x)$. \label{line:remove}
    \EndFor
\EndIf

\EndWhile

\State Let $\widehat{S} \gets T$. 
\State \textbf{return} $\widehat{S}$  \label{line:return}

\end{algorithmic}
\end{algorithm}

\subsection{Certification of Solutions} \label{sec:cert}

In this section, we state the certificate lemma (from \cite{NieGS24}) that provides a way to bound the $W_{1,k'}$ distance between the current filtered version of the dataset $T$ and the uniform distribution over the original inliers $S_0$.
This bound is expressed as a function of a variance-like quantity of the dataset. 
This insight informs the design of our algorithm's stopping condition (cf. \Cref{line:stopping_cond}). Consequently, we can guarantee that upon termination, if the variance is sufficiently small, the solution output by the algorithm will be close to the uniform distribution over the inliers $S_0$ in the $W_{1,k'}$ metric.

\begin{lemma}[Lemma 20 in \cite{NieGS24}]\label{lem:certificate_colt}
    Let $S_0$ be an $(\eps,\delta)$-stable set.
    Let $S'$ be any set satisfying $W_{2,k'}(S',S_0) \leq r$ and 
     $T$ be a set with $|T\cap  S'| \geq (1-\eps)|T|$.
    Then, the following holds:\footnote{ The following result is implied by \cite[Lemma 20]{NieGS24} after using $P'=P$ both being equal to  the uniform distribution over $S'$, and $Q$ being the uniform distribution over $T$.}
    \begin{align*}
        W_{1,k'}(T,S') \lesssim \delta \sqrt{k'}  + r \sqrt{\eps} + \eps \sqrt{r'}\;,
    \end{align*}
    where $r' = \sup_{\bM \in \cM_{k'}} \langle \bM, \vec \Sigma_{T \setminus  S'} \rangle + (\mu_{T \setminus  S'} - \mu_{S'})^\top \bM (\mu_{T \setminus  S'} - \mu_{S'})$.
\end{lemma}

The next result  provides an upper bound for the quantity $r'$ appearing in \Cref{lem:certificate_colt} in terms of the simpler quantity $\langle \bM, \vec \Sigma_T\rangle - k$. 
This simpler quantity acts as the stopping condition of our algorithm (\Cref{line:stopping_cond}). 
In addition, the lemma below also bounds $\| \mu_T - \mu \|_\bM$, which is the error of the empirical mean over the current dataset $T$.
\begin{lemma}\label{lem:certificate_classic}
Let $k$ be a positive integer,  $\bM \in \cM_k$,
and $\tilde \delta \geq \eps$.
    Let $S'$ be a set satisfying the $(\eps,\tilde \delta,k)$-generalized-stability with respect to the vector $\mu \in \R^d$ (cf. \Cref{def:condition}).
    Let $T$ be a set such that $|T \cap S'| \geq (1-\eps)|T|$ and denote $\lambda := \langle \bM, \vec \Sigma_T\rangle - k$.
    Then, the following hold:\footnote{Recall that $\|x\|_\bM = \sqrt{x^\top \bM x}$ denotes the Mahalanobis norm of $x$ with respect to $\bM$.}
    \begin{enumerate}
        \item $\| \mu_T - \mu \|_\bM \lesssim  \tilde \delta +\eps \sqrt{k} + \sqrt{\lambda \eps}$.
        \item $\max\left(\langle \bM, \vec \Sigma_{T \setminus S'} \rangle, \|\mu_{S' \cap T} - \mu_{T \setminus S'}\|_{\bM}^2 \right)  \lesssim \lambda/\eps +\tilde{\delta}^2/\eps^2  +  k$.
    \end{enumerate}
\end{lemma}
\begin{proof}
    The covariance matrix can be decomposed as follows: 
    \begin{align*}
        \vec \Sigma_T = (1-\eps)\vec \Sigma_{S' \cap T} + \eps \vec \Sigma_{T \setminus S'} + \eps(1-\eps) (\mu_{S' \cap T} - \mu_{T\setminus S'})(\mu_{S' \cap T} - \mu_{T\setminus S'})^\top \;.
    \end{align*}
    Using the decomposition above with our assumptions, we obtain
    \begin{align}
        k + \lambda &\geq \langle \bM, \vec \Sigma_T \rangle \notag\\
        &= (1-\eps) \langle \bM, \vec \Sigma_{S' \cap T}  \rangle + \eps\langle \bM, \vec \Sigma_{T \setminus S'} \rangle + \eps(1-\eps) (\mu_{S' \cap T} - \mu_{T \setminus S'})^\top \bM (\mu_{S' \cap T} - \mu_{T \setminus S'}) \notag\\
        &\geq (1-\eps) \left( k - O\left(\frac{\tilde \delta^2}{\eps}\right) \right) + \eps\langle \bM, \vec \Sigma_{T \setminus S'} \rangle + \eps(1-\eps) (\mu_{S' \cap T} - \mu_{T \setminus S'})^\top \bM (\mu_{S' \cap T} - \mu_{T \setminus S'}) \;, \label{eq:ineq}
    \end{align}
    where the second line is derived by the assumption that $S'$ is a set satisfying the $(\eps,\tilde \delta,k)$-generalized-stability as follows:
    \begin{align*}
      \langle \bM, \vec \Sigma_{S' \cap T}  \rangle &= \left\langle \bM, \frac{1}{|S' \cap T|}\sum_{x \in S' \cap T}(x - \mu_{S' \cap T})(x - \mu_{S' \cap T})^\top  \right\rangle  \\
      &= \left\langle  \bM, \frac{1}{|S' \cap T|}\sum_{x \in S' \cap T}(x - \mu )(x -\mu)^\top \right\rangle + \left\langle \bM, (\mu - \mu_{S' \cap T})(\mu - \mu_{S' \cap T})^\top \right\rangle \\&+ \left\langle \bM, \frac{1}{|S' \cap T|} \sum_{x \in S' \cap T}(x - \mu )( \mu- \mu_{S' \cap T})^\top \right\rangle + \left\langle \bM, \frac{1}{|S' \cap T|} \sum_{x \in S' \cap T}( \mu- \mu_{S' \cap T}) (x - \mu )^\top \right\rangle   \\
      &\geq \left\langle  \bM, \frac{1}{|S' \cap T|}\sum_{x \in S' \cap T}(x - \mu )(x -\mu)^\top \right\rangle -  (\mu - \mu_{S' \cap T})^\top \bM (\mu - \mu_{S' \cap T}) \\
      &\geq k - \frac{\tilde \delta^2}{\eps} - 2\|\bM\|_\op \|\mu - \mu_{S' \cap T}\|_2 \geq k - \frac{\tilde \delta^2}{\eps} - 2\tilde \delta \geq  k -   \frac{3\tilde \delta^2}{\eps}  \;,
    \end{align*}
 where we used that $\bM \preceq \bI$, $\|\mu - \mu_{S' \cap T}\|_2 \lesssim \tilde \delta$ by the stability assumption for $S'$ and $\tilde \delta \geq \eps$.
Rearranging \eqref{eq:ineq} yields
   \begin{align}\label{eq:quadratic}
       \langle \bM, \vec \Sigma_{T \setminus S'} \rangle + (1-\eps) (\mu_{S' \cap T} - \mu_{T \setminus S'})^\top \bM (\mu_{S'\cap T} - \mu_{T \setminus S'}) \lesssim  \lambda/\eps +\tilde{\delta}^2/\eps^2  +  k\;.
   \end{align}
   This implies that both $\langle \bM, \vec \Sigma_{T \setminus S'} \rangle$ and $(\mu_{S' \cap T} - \mu_{T \setminus S'})^\top \bM (\mu_{S' \cap T} - \mu_{T \setminus S'})$ are at most $O(\lambda/\eps +\tilde{\delta}^2/\eps^2  +  k)$, which shows the second part of our lemma. 
   The first part of \Cref{lem:certificate_classic} follows simply by the decomposition below:
   \begin{align*}
      \|\mu_{T} - \mu \|_\bM
      &= \|(1-\eps) \mu_{S' \cap T} + \eps \mu_{T \setminus S'} -\mu \|_\bM\\
      &\leq \| \mu_{S' \cap T} - \mu \|_\bM + \eps \| \mu_{S' \cap T} - \mu_{T \setminus S'}\|_\bM
      \\
      &\lesssim \tilde \delta +\eps \sqrt{k} + \sqrt{\lambda \eps} \;.
   \end{align*}
   where we used the generalized-stability assumption to bound the first term, and $(\mu_{S' \cap T} - \mu_{T \setminus S'})^\top \bM (\mu_{S' \cap T} - \mu_{T \setminus S'}) = O(\lambda/\eps +\tilde{\delta}^2/\eps^2  +  k)$ from earlier to bound the second term.

\end{proof}

\subsection{Filtering Scores}\label{sec:distr-learning-proof}
In this subsection, we show that the scores $\tau(x)$ used in \Cref{alg:mean_estimation}
remove more outliers than inliers in expectation in each round.

\begin{lemma}[Analysis of One Round of Filtering: Scores of Outliers $>$ Scores of Inliers]\label{lem:one_round}
    Let $S'$ be a multi-set of $\R^d$ satisfying $(\eps,\tilde \delta,k)$-generalized-stability with respect to $\mu \in \R^d$. Assume the following: $\eps \in (0,c)$ for a sufficiently small absolute constant $c$,  and $\tilde \delta \geq \eps \sqrt{k}$. 
    Let $T$ be a multiset such that $|T\cap S'| \geq (1 - 20\eps)$.
    Let $C$ be a sufficiently large absolute constant. Let $\tau(x)$ be the scores as defined in Line \ref{line:scores_def} of \Cref{alg:mean_estimation}, i.e., $g(x):=\|x - \mu_T\|_\bM^2$, $L$ is the set of points in $T$ with the $\eps \cdot |T|$ largest scores $g(x)$, and $\tau(x):=g(x) \1_{x \in L}$.
    If $\bM \in \cM_{k}$ is a matrix with $\langle \bM, \vec \Sigma_T \rangle > k + C \tilde{\delta^2}/\eps$, then $\sum_{x \in S' \cap T}\tau(x) \leq 0.1 \sum_{x \in T}\tau(x)$.
\end{lemma}

\begin{proof}

Denote $\lambda := \langle \bM,\vec \Sigma_T\rangle - k$. 
For the inlier points, we have the following (explanations are provided after the inequalities):
\begin{align}
    \sum_{x \in S' \cap T} \tau(x) &= \sum_{x \in S' \cap  L} g(x)  = \sum_{x \in S' \cap  L} \|x - \mu_T\|_\bM^2 \label{eq:stepp_1}\\ 
    &\leq 2 \sum_{x \in S' \cap L}(x-\mu)^\top \bM (x-\mu) + 2 |S' \cap L| \cdot \|\mu - \mu_T\|_\bM^2 \label{eq:stepp_2}\\
    &\lesssim |S'| \tilde{\delta^2}/\eps + \eps  |S'|(\tilde \delta^2 + \eps^2 k+ \eps\lambda) \leq 0.01 \lambda |S'| \;,\label{eq:stepp_3}
\end{align}
where the steps used were the following:
\eqref{eq:stepp_2} used the triangle inequality combined with the inequality $(a+b)^2 \leq 2 a^2 + 2 b^2$,
the first term in \eqref{eq:stepp_3} was bounded using the third condition in \Cref{def:condition_3} (and our assumption that $S'$ satisfies the generalized stability condition), the second term in \eqref{eq:stepp_3} used that  $\|\mu - \mu_T\|_\bM^2 \lesssim \tilde \delta^2 + \eps^2 k+ \eps\lambda$ by the certificate lemma (\Cref{lem:certificate_classic}), and we also used that $|S' \cap L| \leq |L| =  \eps |T| \leq \frac{\eps}{1-\Omega(\eps)}|S'|$. The last inequality in \eqref{eq:stepp_3} uses our assumptions $\lambda \geq C \tilde \delta^2/\eps$, $\tilde \delta \geq \eps$, $\tilde \delta \geq \eps \sqrt{k}$, $C \gg 1$, $\eps \ll 1$.

We now show the lower bound for the sum of the scores over all points:
\begin{align}
    \sum_{x \in T} \tau(x) = \sum_{x \in T \cap L} g(x) \geq \sum_{x \in T \setminus S'} g(x)
    \geq \sum_{x \in T} g(x) - \sum_{x \in S' \cap T} g(x) \;. \label{eq:allpoints}
\end{align}
where the second step above is based on the fact that $|T \setminus S'| \leq \eps |T|$ and that $L$ is defined to be the points with the largest $\eps |T|$ scores.
For the first term, we have that, by definition:
\begin{align}
    \sum_{x \in T} g(x) = \langle \bM, \vec \Sigma_T \rangle = (k + \lambda)|T| \geq (k + \lambda)(1-\eps)|S'| \;. \label{eq:lower_bound_all}
\end{align}
For the second term in the RHS of \eqref{eq:allpoints}, we have the following:
\begin{align}
    \sum_{x \in S' \cap T} g(x)  &\leq \sum_{x \in S' } g(x) 
    = \sum_{x \in S'} (x - \mu_T)^\top \bM (x - \mu_T) \notag \\
    &= \sum_{x \in S'} (x - \mu)^\top \bM (x - \mu)
    + |S'| (\mu - \mu_T)^\top \bM (\mu - \mu_T)
    + 2 \sum_{x \in S'} (x - \mu)^\top \bM (\mu - \mu_T) \;. \label{eq:upper_bound_inliers_1}
\end{align}
The first term is at most $(k + \tilde \delta^2/\eps)|S'|$ by our generalized-stability assumption.
The second term is at most $|S'| (\tilde \delta^2 + \eps^2 k + \eps \lambda)$ by \Cref{lem:certificate_classic}.
For the third term, we have that
\begin{align}
    \frac{1}{|S'|} \sum_{x \in S'} (x - \mu_T)^\top \bM (\mu - \mu_T) 
    = (\mu_{S'} - \mu_T)^\top \bM (\mu - \mu_T)
    \leq\| \mu_{S'} - \mu_T\|_2 \|\bM\|_\op \| \mu - \mu_T\|_2&
    \notag\\
    \lesssim (\|\mu - \mu_{S'}\|_2 + \|\mu-\mu_T\|_2)\| \mu - \mu_T\|_2
    \lesssim \tilde \delta^2 + \eps^2 k +  \lambda \eps \;,\label{eq:upper_bound_inliers_2}
\end{align}
where we used that $\| \bM\|_\op \leq 1$, and then we applied  the triangle inequality and \Cref{lem:certificate_classic}.
Putting \eqref{eq:allpoints}-\eqref{eq:upper_bound_inliers_2} together, we have that
\begin{align*}
    \sum_{x \in T} \tau(x) &\geq   (k + \lambda)(1-\eps)|S'|  - (k+ \tilde \delta^2 +\tilde \delta^2/\eps + \eps^2 k +  \lambda \eps)|S'| \notag\\
    &\geq (k + \lambda)(1-\eps)|S'|  - (k+ 0.001 \lambda)|S'| \tag{using $\lambda \geq C \tilde \delta^2/\eps$, $\tilde \delta \geq \eps$, $\tilde \delta \geq \eps \sqrt{k}$, $C \gg 1$, $\eps \ll 1$}\\
    &\geq (\lambda - \eps k - \lambda \eps - 0.001\lambda)|S'|\\
    &\geq 0.9 \lambda |S'| \;. \tag{using  $\eps \ll 1$, $\lambda \geq C \tilde \delta^2/\eps \geq C \eps k$}
\end{align*}
Combining with \eqref{eq:stepp_3} concludes the proof of this lemma.

\end{proof}

\subsection{Proof of \Cref{thm:main_distr_learning}}
We are now ready to combine the previous components to conclude the analysis of our algorithm and complete the proof of \Cref{thm:main_distr_learning}.

\begin{proof}[Proof of \Cref{thm:main_distr_learning}]
We use $k'=k$ without loss of generality, as the same arguments go through for any other $k' \leq k$ by noting that the local corruptions adversary is stronger for $k' \leq k$.

We briefly recall the notation. As in the theorem statement, $S_0 = \{x_1,\ldots,x_n \}$ is the original set of inliers (before any kind of corruptions), which is assumed to satisfy the stability conditions, $S = \{x_i + \Delta_i\}_{i \in [n]}$ is the set after the strong Wasserstein corruptions of \Cref{def:strongWass-k}, ($\Delta_i$'s denote the shift that each point undergoes), and $T$ is the final dataset after globally corrupting $S$ (\Cref{def:outliers}).

The set $S_0 = \{x_1,\ldots,x_n \}$ is $(\eps,\delta)$-stable by assumption. Using \Cref{lem:stability-bigger-param}  it also satisfies $(\eps,\delta',k)$-generalized stability, with $\delta' \lesssim \sqrt{k}\delta$.  By \Cref{prop:subset-covariance-general-k} we have the existence of a set $S_0' \subset S_0$ with $|S_0'| \geq (1-\eps)|S_0|$ such that $\max_{\bM \in \cM_{k}}\tfrac{1}{|S_0'|} \sum_{i : x_i \in S_0'}\|\Delta_i\|_{\bM}^2 \lesssim \rho^2/\eps$. By \Cref{lem:average_roots} we have that the same set $S_0'$ also satisfies $\max_{\bM \in \cM_{k}}\tfrac{1}{|S_0'|} \sum_{i : x_i \in S_0'}\|\Delta_i\|_{\bM} \lesssim \max_{\bV \in \cV_{k}}\tfrac{1}{|S_0'|} \sum_{i : x_i \in S_0'}\|\Delta_i\|_{\bV} \lesssim \rho$ (where the last inequality is by definition of our local perturbation model).

By the above discussion \Cref{lem:existence-of-stability-and-W_1-and-W-2-new} is applicable, stating that if $S'$ is defined to be $S' \eqdef \{x_i + \Delta_i : x_i \in S_0' \}$, then  $S'$ satisfies $(\eps,\tilde \delta , k)$-generalized-stability with respect to $\mu$, for $\tilde \delta = O(\delta \sqrt{k} + \rho)$.
Moreover,   for any $(1-\eps)|S'|$ sized subset $S''$ of $S'$ it holds  that 
\begin{align}
\nonumber
 W_{1,k} (S_0',S'') &\lesssim \rho + \eps\sqrt{k} + \tilde \delta \lesssim \rho + \delta \sqrt{k} \\
 \text{and } W_{2,k} (S_0', S'') &\lesssim \sqrt{\eps k} +  \tilde \delta/\sqrt{\eps} + \rho/\sqrt{\eps} \lesssim (\sqrt{k} \delta + \rho)/\sqrt{\eps}, \numberthis \label{eq:certifiacte-wassertein-subset}    
\end{align}
where we used that $\delta \geq \eps$.

We now argue that filtering does not remove too many ``stable inliers''  ($T \cap S'$) throughout its execution. 
\Cref{lem:one_round} states that, as long as the main while loop of the algorithm has not been terminated, the scores $\tau(x)$ that the algorithm assigns to inlier points in $T$ is substantially bigger (at least by a constant factor) than the ones for outlier points.
Following the standard analysis of filtering algorithms in \cite{DiaKan22-book}, we obtain that 
with probability at least $9/10$, we have that $|T \triangle S'| \leq 20\eps$ throughout the execution.

Let $\widehat{S}$ denote the set $T$ at the end of the filtering algorithm and condition on the high probability event that satisfies $|\widehat{S} \cap S'| \leq 20 \eps$ and $\langle \bM, \vec \Sigma_{\widehat{S}}\rangle \leq k + C \widetilde{\delta}^2/\eps$ for all $\bM \in \cM_k$.
To apply \Cref{lem:certificate_colt}, we need a bound on $W_{2,k'}(S_0,S')$, which we obtain below:
\begin{align*}
W_{2,k}(S_0, S') &\leq W_{2,k}(S_0, S_0') + W_{2,k}(S_0', S') \leq \sqrt{k \eps} + \delta \sqrt{k/\eps} + \sqrt{\rho^2/\eps},
\end{align*}
where use \Cref{w1-w2-subsets-prelims} for the first term and \eqref{eq:certifiacte-wassertein-subset} for the second term (with $S''=S'$).
With this bound on $W_{2,k}(S', S_0')$, applying \Cref{lem:certificate_colt} yields
\begin{align*}
    W_{1,k}(\widehat{S},S') \lesssim \delta \sqrt{k} + \rho  + \eps \sqrt{r'}\;,
\end{align*}
where $r'$ is defined in \Cref{lem:certificate_colt}. To upper bound $r'$, we apply \Cref{lem:certificate_classic} with $T =\widehat{S}$ and $\eps ' = 20\eps$ in place of the parameter $\eps$ appearing in the statement of that lemma. This gives that 
\begin{align*}
    r' &= \sup_{\bM \in \cM_k}\langle \bM, \vec \Sigma_{\widehat{S} \setminus S'}\rangle + \| \mu_{\widehat{S} \setminus S'} - \mu_{S'} \|_\bM^2 
    \\
    &\lesssim \sup_{\bM \in \cM_k} \langle \bM, \vec \Sigma_{\widehat{S} \setminus S'}\rangle + \| \mu_{\widehat{S} \setminus S'} - \mu_{S' \cap \widehat{S}} \|_\bM^2  + \| \mu_{S' \cap \widehat{S}} - \mu_{S'} \|_\bM^2\\
    &\lesssim \tilde \delta^2/\eps^2 +  k  + \| \mu_{S' \cap \widehat{S}} - \mu_{S'} \|_\bM^2\\
    &\lesssim \tilde \delta^2/\eps^2 \;.
\end{align*}
where the last line used $\| \mu_{S' \cap \widehat{S}} - \mu_{S'} \|_\bM^2 \lesssim \| \mu_{S' \cap \widehat{S}} - \mu \|_\bM^2 + \| \mu - \mu_{S'} \|_\bM^2 \lesssim \tilde \delta^2/\eps^2$ by the generalized stability of $S'$ (we also used that $\tilde \delta \geq \eps$ and $\tilde \delta \geq \eps \sqrt{k}$).
Plugging this back, we obtain a bound of $W_{1,k}(\widehat{S}, S') \lesssim \delta \sqrt{k} + \rho$.
We can translate this into a bound for $W_{1,k}(S_0, \widehat{S})$ using the triangle inequality as follows:
\begin{align*}
    W_{1,k}(S_0, \widehat{S}) \leq W_{1,k}(S_0, S_0') + W_{1,k}(S_0', S') + W_{1,k}(S', \widehat{S}). 
\end{align*}
The first term above is upper bounded by  $ \eps \sqrt{k} + \delta $ by \Cref{w1-w2-subsets-prelims},  the second term is upper bounded by $\rho$ by the definition of local contamination, and the last term was shown to be at most $\delta \sqrt{k} + \rho$.
Combining these three terms yields the desired result.

\paragraph{Runtime} 
Note that the algorithm removes at least one point per iteration and the set $S'\cap T$ satisfies the stopping condition of \Cref{line:stopping_cond}.
This is because of stability of $S'$ and therefore stability of any large subset of $S'$ (cf. \Cref{lem:eps-stability-to-2eps}). This means that the algorithm will terminate after $O(n)$ iterations.
In each iteration, the algorithm requires solving an SDP (\Cref{line:optimization-problem}), which can be done in polynomial time by ellipsoid method or interior point method~\cite{Nesterov04}.
\end{proof}

\section{Principal Component Analysis}
\label{sec:pca}
In this section, we present our result for robust principal component analsysis (PCA) in \Cref{thm:PCA} below. 
The goal for PCA is to output a high-variance direction $v$ 
of the unknown covariance $\vec \Sigma$ in the following sense: $v^\top \vec \Sigma v \geq (1-\gamma)\|\vec \Sigma\|_\op$ for $\gamma$ as small as possible. 
Under the global contamination model, robust PCA algorithms have been developed in \cite{KonSKO20,JamLT20,DiaKPP23-pca,JamKLPPT24}.

We will work with zero-mean distributions (for inliers) in this section, which is without loss of generality because one can always reduce to this setting by taking differences of pairs of samples. 

We state an appropriate version of the stability condition which is more relevant to PCA.
\begin{definition}[PCA Stability]\label{def:stabilityPCA}
    Let $0< \eps \leq \gamma$.
A finite multiset $S \subset \R^d$ is called $(\eps,\gamma)$-PCA-stable with respect to a PSD matrix $\vec \Sigma \in \R^{d \times d}$ if for
every $S' \subseteq S$ with $|S'| \geq (1-\eps)|S|$, the following  holds: $(1-\gamma) \vec \Sigma \preceq \frac{1}{|S'|} \sum_{x \in S'} x x^\top \preceq (1+ \gamma) \vec \Sigma$. 
\end{definition}
The definition above is very closely related to \Cref{def:condition} as shown by the following observation.
    \begin{fact}\label{fact:normalization}
    Let $\vec \Sigma$ be a positive definite matrix.
    If a set of samples $\{\vec \Sigma^{-1/2} x_i\}_{i \in [n]}$ is $(\eps,\delta)$-stable (\Cref{def:stability1}) with respect to $\mu=0$ (\Cref{def:condition}), then $\{x_i \}_{i \in [n]}$ is $(\eps,\gamma)$-PCA-stable with respect to $\vec \Sigma$ (\Cref{def:stabilityPCA}) for $\gamma \lesssim \delta^2/\eps$.
    \end{fact}
Using the connection above, it can be seen that the stability definition above is satisfied by many distribution families of interest.
Similarly, a set of i.i.d.\ samples from such distributions continue to satisfy this definition with high probability~\cite{JamLT20,JamKLPPT24}.
Consequently, the stability-based algorithms obtain the state of the art results for  
robust PCA for many distribution families~\cite{JamLT20,DiaKPP23-pca,JamKLPPT24}.

\begin{definition}[Stability-based Algorithms for PCA]\label{def:stability_based_PCA}
Let $S$ be an $(\eps,\gamma)$-PCA-stable set with respect to an (unknown) PSD matrix $\vec \Sigma \in \R^{d\times d}$ (\Cref{def:stabilityPCA}).
Let $T$ be any set in  $\cO(S,\eps)$ (cf. \Cref{def:outliers}).
We call an algorithm stability-based PCA-algorithm if it takes as an input $T$, $\eps$, and $\gamma$,
and outputs a unit vector $v \in \R^d$ in polynomial time such that $v^\top \vec \Sigma v \geq (1 - O(\gamma)) \|\vec \Sigma\|_{\op}$.
\end{definition}

For this section, we consider a version of the contamination model where local corruptions are introduced to the data after whitening.
\begin{contModel}[Strong Local Contamination After Whitening]
\label{def:strongWass-whitenned}
Let $\overline \rho \geq 0$.
    Let $S_0=\{x_1,\ldots,x_n \}$ be an $n$-sized set in $\R^d$ and $\vec \Sigma \in \R^{d \times d}$ be a PSD matrix. 
    Consider an adversary that perturbs each point $x_i$ to $\widetilde{x}_i$ with the only restriction that in each direction, the average perturbation is at most $\overline \rho$. Formally, we define 
    \begin{align}
\cW^\strong(S_0,\overline\rho,\vec \Sigma) := \left\{S= \{\widetilde{x}_1,\dots, \widetilde{x}_n\} \subset \R^d :     \sup_{v \in \R^d: \|v\|_2=1}\,\,\frac{1}{n} \sum_{i \in [n]} \big|v^\top \vec \Sigma^{-1/2}(\widetilde{x}_i - x_i) \big| \leq \overline \rho    \right\}\,. \notag
    \end{align}
    The adversary returns an arbitrary set  $S \in \cW^\strong(S_0,\overline \rho, \vec \Sigma)$ after possibly reordering the points.
\end{contModel}
 As a remark, we note that \Cref{def:strongWass-whitenned} and \Cref{def:strongWass-intro} are equivalent to each other, as long as the matrix $\vec \Sigma$ is well-conditioned. In particular, $\overline \rho \leq  \rho/\sqrt{\lambda_{\min}}$ and $\rho \leq \overline \rho \sqrt{\lambda_{\max}}$ where $\lambda_{\max},\lambda_{\min}$ denote the largest and smallest eigenvalues of $\vec \Sigma$ respectively.
The appealing property of the whitened local perturbations is that (i) it allows the amount of local perturbations to increase in high-variance directions, and (ii) it decouples the local contamination parameter $\overline \rho$ from the scale of the covariance matrix $\vec\Sigma$.\footnote{Indeed, it can be seen that the range of parameter $\rho$ where robust PCA is non-trivial depends on the scale of $\vec \Sigma$. For example, consider the case when the inlier distribution is $\cN(0, \sigma^2(\bI + vv^\top))$ for a unit vector $v$ and a scalar $\sigma^2$. Then the $2$-Wasserstein distance between $\cN(0,\sigma^2(\bI + vv^\top))$ and $\cN(0,\sigma^2\bI )$ is $\Theta(\sigma)$. 
Hence, for $\rho \gtrsim \sigma$, the local adversary can simulate samples from an isotropic distribution, removing any signal from the direction of interest $v$.
On the other hand, as shown in \Cref{thm:PCA}, the range of $\overline{\rho}$ does not depend on $\vec \Sigma$.
}

\begin{theorem}\label{thm:PCA}
    Let $c$ be a sufficiently small positive constant and $C$ a sufficiently large constant.
    Let outlier rate $\eps \in (0,c)$ and  contamination radius $\overline \rho \in (0,\sqrt{\eps})$.
    Let $S_0$ be a set of samples satisfying 
    $(\eps,\gamma)$-PCA-stability with respect to a PSD matrix $\vec \Sigma \in \R^{d \times d}$ (\Cref{def:stabilityPCA}). 
    Let $T$ be a corrupted dataset after  $\eps$-fraction of outliers and $\overline \rho$-strong local corruptions after whitening (as per \Cref{def:strongWass-whitenned}).
    Then, any stability-based PCA algorithm (\Cref{def:stability_based_PCA}) on input $T, \eps,\tilde \gamma = C \cdot (\gamma + \frac{\overline \rho^2}{\eps})$, outputs a unit vector  $v$ such that with high probability (over the internal randomness of the algorithm):
        $v^\top \vec \Sigma v \geq \left(1 - O\left(\gamma  + \frac{\overline \rho^2}{\eps}\right)\right) \| \vec \Sigma \|_\op$.  
\end{theorem}
\begin{proof}
    \looseness=-1 For a matrix $\bA\in \R^{d\times d}$ and set $S \subset \R^d$, we use $\bA [S]$ to denote the set $\{\bA x: x \in S\}$.
    Let $S$ be the set after the local perturbations of $S_0$ (as per \Cref{def:strongWass-whitenned}).
    It suffices to show that $S$ contains a subset $S' \subset S$ such that $|S'|\geq (1-\eps)$ and $S'$ is $(\eps,\widetilde \gamma)$-PCA stable with respect to $\vec \Sigma$ for $\tilde\gamma \lesssim \gamma + \overline{\rho}^2/\eps$.
    By \Cref{fact:normalization}, it suffices to show that the whitened data $\vec \Sigma^{-1/2} [S]$ contains a large subset $S'$ that is $(\eps,\tilde \gamma)$-PCA-stable
    with respect to $\bI$. 
    Leveraging the connections between PCA-stability and the usual stability (\Cref{def:condition}), it suffices to show that $S'$ satisfies the conditions pertaining to the second moment of $(\eps,\sqrt{\eps \gamma} + \overline \rho)$-stability (with respect to $\mu=0$).
    The existence of a large $S'$ with the desired stability can be shown by following the proof in \Cref{lem:existence-of-stability-and-W_1-and-W-2-new} mutatis mutandis for $k=1$ and $\delta=\sqrt{\eps \gamma}$.%
    \footnote{In fact, if we make the stronger assumption in the theorem that $\vec \Sigma^{-1/2} [S_0]$ is $(\eps,\delta,1)$-generalized stable (as opposed to PCA stable), then the desired conclusion follows as a direct corollary from \Cref{lem:existence-of-stability-and-W_1-and-W-2-new} and \Cref{fact:normalization}.}

\end{proof}

Finally, we briefly mention how to generalize \Cref{thm:PCA} to $k$-robust PCA for $k>1$.
Observe that in the proof above,  we have shown that $S$ contains a large subset that is $(\eps, \widetilde{\gamma})$-PCA-stable for $\widetilde{\gamma} \lesssim \gamma + \overline{\rho}^2/\eps$.
Generalization to $k>1$ then follows directly from \cite[Corollary 3]{JamKLPPT24}.

\section{ Sum-of-squares Based Algorithm: Proof of \Cref{thm:sos-mean}}
\label{sec:sos}

In this section, we prove the result on mean estimation under the combined contamination model for distributions with certifiably bounded moments in the sum-of-squares (SoS) proof system. We show that the approach of \cite{KotSS18,HopLi18} extends to the contamination model of this paper.

We refer the reader to \cite{BarSte16-sos-notes,FleKP19-sos} for the necessary definitions of the terms such as degree-$d$ SoS proofs and pseudoexpectations. 
We list only a few basic facts that we use and refer the reader to the aforementioned references for the full background.

\begin{fact}[Cauchy-Schwarz for Pseudoexpectations]\label{fact:CS_pseudo_exp}
	Let $f,g$ be polynomials of degree at most $t$. Then, for any degree-$2t$ pseudoexpectation $\pE$,
	$\pE[fg] \leq \sqrt{\pE [f^2]} \sqrt{\pE[g^2]}$.
	Consequently, for every squared polynomial $p$ of degree $t$, and $k$ a power of two, 
	$\pE[p^k] \geq  (\pE[p])^k$ for every $\pE$ of degree-$2tk$.
	\label{fact:pseudo-expectation-cauchy-schwarz}
\end{fact} 

\begin{fact}[SoS Triangle Inequality]\label{fact:sos-triangle}
	If $k$ is a power of two, 
	$\sststile{k}{a_1, a_2, \ldots, a_n} \left\{ \left(\sum_i a_i \right)^k \leq n^k \left(\sum_i a_i^k\right) \right\}.$ 
\end{fact}

\begin{fact}[SoS Cauchy-Schwartz and H\"older]\label{fact:sos-holder}
	Let $f_1,g_1,  \ldots, f_n, g_n$ be indeterminates
	over $\R$. Then, 
	\begin{align*}
	\sststile{2}{f_1, \ldots, f_n,g_1, \ldots, g_n} \Set{ \left(\frac{1}{n} \sum_{i=1}^n f_i g_i \right)^{2} \leq \left(\frac{1}{n} \sum_{i=1}^n f_i^2\right) \left(\frac{1}{n} \sum_{i=1}^n g_i^2\right) } \;.
	\end{align*} 
	Moveover, if $p_1, \dots, p_n$ are indeterminates, for any $t \in \Z_+$ that is a power of $2$, we have that
	\begin{align*}
	\{w_i^2 = w_i \mid i \in [n] \} \sststile{O(t)}{p_1, \dots, p_n} \left( \sum_i w_i p_i\right)^t &\leq \left( \sum_{i \in [n]} w_i\right)^{t-1} \cdot \sum_{i \in [n]} p_i^t \quad \text{ and} \\
	\{w_i^2 = w_i \mid i \in [n] \} \sststile{O(t)}{p_1, \dots, p_n} \left( \sum_i w_i p_i\right)^t &\leq \left( \sum_{i \in [n]} w_i\right)^{t-1} \cdot \sum_{i \in [n]} w_ip_i^t \;.
	\end{align*} 
\end{fact}

\begin{definition}[Certifiably Bounded Moments]

For an even $t\in \N$, we say a distribution $P$ with mean $\mu_P$ over $\R^d$ has $(t,M)$-certifiably bounded moments if the polynomial inequality $p(v) \geq 0$ for $p(v):= M^t - \E_{X\sim P}[\langle v, X - \mu_P\rangle^t]$
has an SoS proof of degree $O(t)$ under the assumption $\|v\|_2^2=1$. If $S$ is a set of points in $\R^d$, we say that $S$ has $(t,M)$-certifiably bounded moments if the uniform distribution over $S$ satisfies the previous definition.
\end{definition}
Many distribution families, such as rotationally invariant distributions, $t$-wise product distributions with bounded moments, and Poincare distributions are known to be certifiably bounded; see, for example, \cite{KotSS18}.
\begin{theorem}
\label{thm:robust-mean-sos}
Let $\eps \in (0,c)$ for a sufficiently small absolute constant $c$.
Let $S$ be a set of $n$ points in $\R^d$ with (unknown) mean $\mu$.
Further assume that the uniform distribution on $S$  
has $(t,M)$-certifiably bounded moments for $t$ being a power of $2$.
Let $T$ be the version of the dataset $S$ after introducing $\eps$-fraction of global outliers and $\rho$-strong local corruptions (as per \Cref{def:combined-outlier-local-intro}).
Then, there exists an algorithm that takes as input $T, \rho, \epsilon, M $, and $t$,
runs in time $\poly(n^t, d^{t^2})$,
and returns $\widehat{\mu}$ such that with probability at least $0.9$, it holds $\|\widehat{\mu}-\mu\|_2 \lesssim M \epsilon^{1-1/t} + \rho$.
\end{theorem}

Let $S=\{x_1,\dots,x_n\}$ be the original dataset of inliers.
Let $S'=\{x_1',\dots,x_n'\}$ with $x_i' = x_i + z_i$ such that $\forall v\in \cS^{d-1}$: $\E_{i \sim [n]}[|\langle v,z_i\rangle|] \leq \rho $ be the dataset after the local corruptions (we use the notation $\E_{i \sim [n]}$ to denote taking the average over $i\in [n]$, for example, $\E_{i \sim [n]}[x_i] = \tfrac{1}{n}\sum_{i \in [n]} x_i$).
Finally, let $T= \{y_1,\dots,y_n\}$ be the final dataset after the global corruptions, i.e., $T$ is such that for all but $\eps n$ of the points we have $x_i'=y_i$.
Let $\cI \subset [n]$ denote the set of indices such that $x_i'=y_i$.

The algorithm is the following: First, it finds a pseudoexpectation $\pE$ over (i) $d$-dimensional variables $(\widetilde{y}_i)_{i=1}^n, (\widetilde{x}_i)_{i=1}^n, (\widetilde{z}_i)_{i=1}^n$, $\widetilde{\mu}$, (ii) scalar variables $(\widetilde{w}_i)_{i=1}^n$, and (iii) appropriate auxiliary variables,\footnote{These are needed for encoding the constraints \Cref{it:sos-6,it:sos-7}. We refer the reader to \cite{HopLi18,KotSte17} for further details on how to encode these constraints using auxiliary variables.
} under the constraints that $\pE$ satisfies the following set of polynomial (in)equalities for a large enough absolute constant $C$:
\begin{enumerate}[label=(I.\roman*)]
  \item \label[ineq]{it:sos-1} For all $i\in [n]$: $\widetilde{w}_i^2=\widetilde{w}_i$.
  \item \label[ineq]{it:sos-2} For all $i\in [n]$: $\widetilde{w}_i \widetilde{y}_i=\widetilde{w}_i y_i$.
  \item \label[ineq]{it:sos-3} $\sum_{i=1}^n \widetilde{w}_i\geq (1-2\eps)n$.
  \item \label[ineq]{it:sos-4} For all $i\in [n]$: $\widetilde{y}_i =  \widetilde{x}_i + \widetilde{z}_i$.
  \item \label[ineq]{it:sos-5} $\widetilde{\mu} = \frac{1}{n}\sum_{i=1}^n \widetilde{y}_i$.
  \item \label[ineq]{it:sos-6} There exists an SoS proof in the variable $v$ of the  inequality
$\E_{i \sim [n]}[ \langle v, \widetilde{x}_i - \widetilde{\mu} \rangle^t] \leq C^t\left(M^t + \rho^t\right)$ under the constraints $\|v\|_2^2=1$.
  \item \label[ineq]{it:sos-7} There exists an SoS proof in the variable $v$ of the  inequality
$\E_{i \sim [n]} [\langle v, \widetilde{z}_i \rangle^2]  \leq C\rho^2/\eps$ under the constraints $\|v\|_2^2=1$.
\end{enumerate}

\noindent Finally, the algorithm outputs $\widehat{\mu} = \pE[\widetilde{\mu}]$.

\subsection{Proof of \Cref{thm:robust-mean-sos}}

If $z_i := x_i' - x_i$ denote the local perturbations, then by \Cref{prop:subset-covariance-intro}, we know that there exists a subset of indices $\cI' \subset [n]$ with $|\cI'| \geq (1-\eps)n$ such that $\E_{i \sim \cI'}[\langle v,z_i \rangle^2] \lesssim \rho^2/\eps$. Without loss of generality, we can treat the remaining points $i \in [n] \setminus \cI'$ as outliers. This is why we use $1-2\eps$ in the right hand side of \Cref{it:sos-3}. Thus, throughout this proof, we assume that $\E_{i \sim [n]}[\langle v,z_i \rangle^2] \lesssim \rho^2/\eps$ and that we have $2\eps$ of outliers (i.e., the set $\cI$ of indices $i \in [n]$ with $x_i'=y_i$ has size $|\cI| \geq (1-2\eps)n$).

\paragraph{Satisfiability.}
We first argue that the system of polynomial inequalities above is satisfiable.
Recall the notation for $S,x_i,x_i',z_i,y_i,T$ provided after the statement of \Cref{thm:robust-mean-sos}. 
Let $\mu = \tfrac{1}{n}\sum_{i \in [n]}x_i$.
To show satisfiability, we make the following choice of the variables: $\tilde x_i = x_i$ for $i \in \cI$ and $\tilde x_i = \mu$ otherwise, $\tilde y_i =  y_i$ for $i \in \cI $ and $y_i = \mu$ otherwise, $\tilde w_i = \1_{i \in \cI }$, and we choose $\tilde z_i= \tilde y_i- \tilde x_i$ (recall that $\cI$ is the set of indices $i$ such that $x_i'=y_i$).

Under these choices, \Cref{it:sos-2,it:sos-4} are satisfied trivially.
Moreover, the $\widetilde{w}_i$'s satisfy the  \Cref{it:sos-1,it:sos-3} because $|\cI| \geq (1-2\eps)n$.

We now argue that the \Cref{it:sos-6} is also satisfiable.
First, define $\tilde \mu = \tfrac{1}{n}\sum_{i \in [n]} \tilde y_i$ and $\tilde \mu' = \tfrac{1}{n}\sum_{i \in [n]} \tilde x_i$, and observe that 
under the above choices of $\widetilde{y}_i$ and $\widetilde{x}_i$,
we see that $\|\widetilde{\mu} - \widetilde{\mu}'\|_2 \lesssim \rho$; this is because $\|\widetilde{\mu} - \widetilde{\mu}'\|_2 \leq \|\sum_{i \in [n]} z_i/n\| \leq \rho$.
Moreover, by SoS Cauchy-Schwarz inequality,  there exists an $O(t)$-degree SoS proof of the following inequality in the variable $v$ under the constraint $\|v\|_2^2=1$: $\langle v, \widetilde{\mu}' - \widetilde{\mu}\rangle^{t} \lesssim  \|\widetilde{\mu}' - \widetilde{\mu}\|_2^{t} \leq \rho^t$.
Applying the SoS triangle inequality \Cref{fact:sos-triangle}, we obtain that the following inequality has an $O(t)$-degree sum of squares proof in the variable $v$:
\begin{align*}
  \E_{i \sim [n]}[ \langle v, \widetilde{x}_i - \tilde \mu\rangle^t]
  &\leq 2^t\E_{i \sim [n]}[ \langle v, \widetilde{x}_i - \tilde \mu' \rangle^t]
  +  2^t\langle  v, \tilde \mu' -\tilde \mu \rangle^t
  \lesssim 2^t(M^t + \rho^t) \;.
\end{align*}

Thus, \Cref{it:sos-6} is satisfiable.

By \Cref{prop:subset-covariance-intro}, $\widetilde{z}_i$ have bounded covariance, which is equivalent to a degree two polynomial inequality in the variable $v$, and since it is a degree-two polynomial, it also has a  sum of squares proof in the variable $v$, 
satisfying \Cref{it:sos-7}.\footnote{Formally, we can replace this constraint by an equivalent constraint,  $\E_{i \sim [n]} \widetilde{z}_i\widetilde{z}_i^\top = (\rho^2/\eps) I - BB^\top$, for some auxiliary matrix variable $B \in \R^{d \times d}$.}
Therefore, all the constraints in our program are satisfied by this construction.

\paragraph{Correctness.}
Fix a  direction $v \in \cS^{d-1}$. Let $\mu = \tfrac{1}{n}\sum_{i \in [n]} x_i$.
We will show that $\langle \widehat{\mu} - \mu, v\rangle = \pE[\langle \widetilde{\mu} - \mu, v\rangle] \leq \tau$ for $\tau = O(M\eps^{1-1/t})$ for any pseudoexpectation $\pE$ satisfying the program  \Cref{it:sos-1,it:sos-2,it:sos-3,it:sos-4,it:sos-5,it:sos-6,it:sos-7}.
By duality between pseudoexpectations and SoS proofs, it suffices to show that there is an SoS proof of $\langle \widetilde{\mu} - \mu, v\rangle \leq \tau$ under the polynomial constraints above.

Let $r_i = \1_{i \in \cI}$ be the locally corrupted inliers
and $W_i=\widetilde{w}_i r_i$ be the variables corresponding to the surviving inliers  ``selected'' by the program.
Let $W_i' = (1-W_i)$.
Then there is an SoS proof of $(W_i')^2=W_i'$ and $\sum_i W_i' \leq 3\eps n$ (with proof similar to Claim 4.3 in \cite{DiaKKPP22-colt}).
Therefore, under the constraints above we have the following (recall that we use the notation $\E_{i \sim [n]}[x_i]$ to denote the average $\tfrac{1}{n}\sum_{i \in [n]} x_i$):
\begin{align*}
\langle v, \widetilde{\mu} - \mu\rangle^{2t}
&=\left(\E_{i \sim [n]}\left[\langle v, \widetilde{y}_i - x_i\rangle\right]\right)^{2t}
=\left(\E_{i \sim [n]}\left[\langle v, \widetilde{y}_i - x_i'\rangle \right]  + \E_{i \sim [n]} \left[ \langle v, z_i\rangle\right]\right)^{2t}\\ 
&\leq 2^{2t }\left(\E_{i \sim [n]}\left[\langle v, \widetilde{y}_i - x_i'\rangle\right]\right)^{2t} + 2^{2t} \left(\E_{i \sim [n]} \left[\langle v, z_i\rangle\right]\right)^{2t} \\
&=2^{2t}\left(\E_{i \sim [n]}\left[W_i'\langle v, \widetilde{y}_i-x_i' \rangle \right]\right)^{2t} + 2^{2t}\left(\E_{i \sim [n]} \left[\langle v, z_i\rangle\right]\right)^{2t} \;.
\end{align*}
where we used the SoS triangle inequality (\Cref{fact:sos-triangle}) and that $\widetilde{y}_i W_i = x_i' W_i$.
The last term, which does not depend on the program variables, is at most $(2\rho)^{2t}$ by assumption.
We now focus on the first term. By using \Cref{it:sos-4} and another application of the SoS triangle inequality:
\begin{align*}
\left(\E_{i \sim [n]}\left[W_i'\langle v, \widetilde{y}_i-x_i' \rangle \right]\right)^{2t} 
&= \left(\E_{i \sim [n]}\left[W_i'\langle v, \widetilde{x}_i + \widetilde{z}_i -x_i' \rangle \right]\right)^{2t} \\
&\leq 2^{2t}\left(\E_{i \sim [n]}\left[W_i'\langle v, \widetilde{x}_i - x_i\rangle\right]\right)^{2t} + 2^{2t}\left(\E_{i \sim [n]}\left[W_i'\langle v,  \widetilde{z}_i - z_i  \rangle \right]\right)^{2t} \;. \numberthis \label{eq:sos122}
\end{align*}

\noindent We shall use different assumptions on $\widetilde{x}_i$ and $\widetilde{z}_i$ to handle these terms differently.
For the first term, we use the SoS H\"older inequality (\Cref{fact:sos-holder}) and the constraints that ${W_i'}^2 = W_i$ with $\E_{i \sim [n]}W_i' \leq 3\eps$ (within the SoS proof system) to get the following:
\begin{align*}
    \left(\E_{i \sim [n]}\left[W_i'\langle v, \widetilde{x}_i - x_i\rangle \right]\right)^{2t}
    &\leq \left(\E_{i \in [n]} [W_i'] \right)^{2t-2}   \left(\E_{i \sim [n]}\left[\langle v, \widetilde{x}_i - x_i\rangle^{t}  \right]\right)^2\\
    &\leq (3\eps)^{2t-2} \left(\E_{i \sim [n]}\left[\langle v, \widetilde{x}_i - x_i\rangle^{t}  \right]\right)^2 \numberthis \label{eq:sos12}
\end{align*}
To control the right hand side above, we further have the following inequalities (in the SoS proof system):
Starting with the SoS triangle inequality, we obtain
\begin{align*}
    \left(\E_{i \sim [n]}\left[\langle v, \widetilde{x}_i - x_i\rangle^{t}  \right]\right)^2
    &\lesssim 2^{t}\left(\E_{i \sim [n]}\left[ \langle v, \widetilde{x}_i - \tilde\mu\rangle^{t} \right]^{2} + \langle v,\tilde \mu - \mu\rangle^{2t}
    + \E_{i \sim [n]}\left[ \langle v,  x_i - \mu \rangle^{t}\right]^{2} \right)\\
    &\lesssim 2^{t}\left( M^{2t} + \rho^{2t} + \langle v,\tilde \mu - \mu\rangle^{2t} + M^{2t}  \right) \;, \numberthis \label{eq:sos11}
\end{align*}
where we used \Cref{it:sos-6} and  the moment assumption on the inliers. Combining \eqref{eq:sos12} and \eqref{eq:sos11} we get that  $\E_{i \sim [n]}\left[(W_i'\langle v, \widetilde{x}_i - x_i\rangle)^{2t}\right]^2 \lesssim (C\eps)^{2t-2}\langle v,\tilde \mu - \mu\rangle^{2t} + (C\eps)^{2t}M^{2t} + (C\rho)^{2t}$.

We now move to the second term in \eqref{eq:sos122}. We apply the SoS H\"older inequality to get
\begin{align*}
\left(\E_{i \sim [n]}\left[W_i'\langle v,  \widetilde{z}_i - z_i  \rangle \right]\right)^2 &\leq 
\left(\E_{i \sim [n]}[W_i']\right)\E_{i \sim [n]}\left[\langle v,  \widetilde{z}_i - z_i  \rangle^2 \right]\\
&\lesssim \eps \left(\E_{i \sim [n]}\left[\langle v,  z_i   \rangle^2 \right] + \E_{i \sim [n]}\left[\langle v,  \widetilde{z}_i   \rangle^2 \right]\right) \\
&\lesssim \eps ( \rho^2/\eps + \rho^2/\eps) = O(\rho^2) \;.
\end{align*}
where the last line uses \Cref{it:sos-7} and our assumption for the inliers.
Combining everything, we have shown an SoS proof of
\begin{align*}
    \langle v, \widetilde{\mu} - \mu\rangle^{2t} \lesssim (C\eps)^{2t-2}\langle v,\tilde \mu - \mu\rangle^{2t} + (C\eps)^{2t-2}M^{2t} + (C\rho)^{2t} \;.
\end{align*}
for some constant absolute $C$.
By solving for $\langle v, \widetilde{\mu} - \mu\rangle^{2t}$, and using that $\eps < c$ for a sufficiently small constant, we get $ \langle v, \widetilde{\mu} - \mu\rangle^{2t} \lesssim  (C\eps)^{2t-2}M^{2t} + (C\rho)^{2t}$.
Finally, by SoS H\"older inequality, we have that $\langle v, \widetilde{\mu} - \mu\rangle \lesssim  (C\eps)^{1-1/t}M + \rho$.

\section*{Acknowledgments}
We thank Po-Ling Loh for helpful comments and discussion on prior work.

\printbibliography

\newpage
\appendix

\section{Additional Preliminaries}\label{sec:additional-prelims}

The statistical rates for stability (\Cref{def:stability1}) are well understood, up to logarithmic factors from the optimal~\cite{DiaKP20}. 

\begin{fact}[Stability rates for nice distribution families~\cite{DiaKP20}]\label{fact:stability-rates}
    Let $\cD$ be a family of distributions.
    Let $c$ be a small enough absolute constant.
    Fix a $D \in \cD$ and let $S_0$ be a set of $n$ samples drawn i.i.d. from  $D$. 
    For each of the cases below and $\eps + \frac{\log(1/\tau)}{n} \in (0,c)$, there exists an $S_1 \subset S_0$ with $|S_1| \geq (1-\eps)|S_0|$
    which is $(\eps,\delta)$-stable with probability at least $1-\tau$, for the following parameter $\delta$:
    \begin{itemize}
        \item If $\mathcal{D}$ is the family of isotropic subgaussian distributions, then $\delta \lesssim \eps \sqrt{\log(1/\eps)} + \sqrt{d/n} + \sqrt{\log(1/\tau)/n}$.
        \item If $\mathcal{D}$ is the family of distributions with isotropic covariance and bounded $k$-th moments, then $\delta \lesssim \eps^{1 - \frac{1}{k}} + \rho+ \sqrt{(d \log d)/n} + \sqrt{\log(1/\tau)/n}$.
        \item If $\mathcal{D}$ is the family of distributions with covariance $\Sigma\preceq I$, then
        $\delta \lesssim \sqrt{\eps} + \sqrt{(d \log d)/n} + \sqrt{\log(1/\tau)/n}$.
    \end{itemize}
\end{fact}

Recall the family of \emph{stability-based algorithms} from  \Cref{def:stability-based-alg}.
Hence,
 we immediately obtain the following, which is a more detailed version of \Cref{fact:ars-existing}.

\begin{fact}\label{fact:mean-est-full}
    Let $\cP$ be a family of distributions.
    Fix a $P \in \cP$ and let $S_0$ be a set of $n$ i.i.d.\ samples from $P$.
    Let $T$ be a corrupted version of $S_0$ with $\eps$-fraction of global outliers (\Cref{def:outliers}).
    Let $\mu$ be the (unknown) mean of $P$.
    There exist computationally-efficient algorithms that take as input $T$, $
    \eps$, and $\rho$ and output $\widehat{\mu}$ with the following guarantees:
    \begin{itemizec}
        \item If $\mathcal{D}$ is the family of isotropic subgaussian distributions, then
        $\|\widehat{\mu}-\mu\|_2 \lesssim \eps \sqrt{\log(1/\eps)}   + \sqrt{d/n} + \sqrt{\log(1/\tau)/n}$.
        \item If $\mathcal{D}$ is the family of distributions with isotropic covariance and bounded $k$-th moments, then
        $\|\widehat{\mu}-\mu\|_2 \lesssim \eps^{1 - \frac{1}{k}} + \sqrt{(d \log d)/n} + \sqrt{\log(1/\tau)/n}$.
        \item If $\mathcal{D}$ is the family of distributions with covariance $\vec \Sigma\preceq \vec I$
        then
        $\|\widehat{\mu}-\mu\|_2 \lesssim \sqrt{\eps} + \sqrt{(d \log d)/n} + \sqrt{\log(1/\tau)/n}$.
    \end{itemizec}
\end{fact}

The above bounds for robust mean estimation are optimal with respect to the parameters $\eps, d, n, \tau$ (up to the $\sqrt{\log d}$ term). Specifically for $\eps$, the matching lower bound can be derived using a simple identifiability argument (see Section 1.3 in \cite{DiaKan22-book} for formal details). 
The (near-)optimality with respect to the remaining parameters are also folklore facts: $\Omega(\sqrt{d/n})$ error is necessary even for estimating the mean of $\cN(\mu,I)$ without outliers (as a standard application of Fano's method), and $\Omega{(\sqrt{\log(1/\tau)/n})}$ is necessary even for one-dimensional $\cN(\mu,1)$ without outliers (see e.g., Proposition 6.1 in \cite{catoni2012challenging}).

\begin{lemma}
    \label{lem:rounding-cont-discrete}
    Let $y_1,\dots,y_n$ be $n$ vectors in $\R^d$ and a $w \in \Delta_{n,\eps}$.
    Let $\mu \in \R^d$ be fixed and define $\mu_{w} := \sum_i w_i y_i$ and $\overline{\vec \Sigma}_w := \sum_i w_i (y_i-\mu)(y_i-\mu)^\top$.
    Then there exists a set $\cI \subset [n]$ satisfying (i) $|\cI| \geq (1-2\eps)n$ and (ii) 
    the set $S:= \{y_i\}_{i \in \cI}$ 
   satisfying   $\max_{\bM \in \cM} \langle \bM, \overline{\vec\Sigma}_S\rangle \lesssim \max_{\bM \in \cM}\langle    \bM, \overline{\vec \Sigma}_w \rangle$
\end{lemma}
\begin{proof}
\looseness=-1Without loss of generality, let $w_1\geq w_2 \geq \dots \geq  w_{n}$.
Define the set $\cI$ to be the set $[(1-2\eps)n]$.
Then for each $i \in \cI$, we have $w_i \gtrsim \frac{1}{n}$. 
Let $S = (y_i)_{i \in \cI}$.
Since $|S| \geq (1-2\eps)n$,
it follows that $\frac{1}{|S|} \lesssim w_i$ for all $i \in \cI$; see \cite[Lemma D.2]{DiaKP20}. Defining $z_i := y_i-\mu$,
for any $\bM \in \cM_k$, we have 
\begin{align*}
\langle \overline{\vec \Sigma}_S, \bM\rangle =    
\frac{1}{|S|} \sum_{i \in \cI} \|z_i\|_{\bM}^2
\lesssim \sum_{i \in \cI} w_i\|z_i\|_{\bM}^2
\lesssim \sum_{i \in [n]} w_i\|z_i\|_{\bM}^2 = \langle \overline{\vec \Sigma}_w , \bM \rangle \lesssim \frac{\rho^2}{\eps}\,.
\end{align*}

\end{proof}

\end{document}